\documentclass[a4paper,11pt]{article}
\usepackage{amssymb,amsmath,amsthm,amsfonts}
\usepackage{xspace}
\usepackage{todonotes}
\usepackage{enumitem}
\usepackage[nolabel]{showlabels}
\usepackage{thmtools}
\usepackage{thm-restate}
\usepackage[hidelinks,pagebackref=true]{hyperref}
\usepackage[nameinlink]{cleveref}
\usepackage{multirow}
\hypersetup{
    colorlinks = true,
    linkcolor = blue,
    citecolor = red 
}
\usepackage{graphicx} 
\usepackage{thmtools,thm-restate}
\usepackage{tabularx}
\usepackage{array,booktabs}
\usepackage{boxedminipage,verbatim,float,caption}
\usepackage[modulo]{lineno}
\usepackage{algorithm}
\usepackage{algpseudocode}
\newtheorem{theorem}{Theorem}[section]
\newtheorem{lemma}[theorem]{Lemma}
\newtheorem{corollary}[theorem]{Corollary}
\newtheorem{claim}[theorem]{Claim}

\newtheorem{observation}[theorem]{Observation}
\theoremstyle{definition}
\newtheorem{definition}[theorem]{Definition}

\newenvironment{claimproof}[1]{\par\noindent\emph{Proof:}\space#1}{{\leavevmode\unskip\penalty9999 \hbox{}\nobreak\hfill\quad\hbox{$\diamondsuit$}}}

\usepackage{tikz}
\usetikzlibrary {arrows.meta,bending,positioning} 
\usepackage[a4paper,bindingoffset=0.2in,
            top=1.1in, bottom=1.1in, left=1.0in, right=1.0in,
            footskip=.5in]{geometry}

\usepackage{xspace}
\newcolumntype{M}[1]{>{\centering\arraybackslash}m{#1}}

\newcommand{\ignore}[1]{}

\newcommand{\pbDef}[3]{
	\noindent
	\vspace*{0.01cm}
	\begin{center}
		\begin{boxedminipage}{0.98 \columnwidth}
			\textsc{#1}\\[1pt]
			\textbf{Input:}  #2\\[0pt]
			\textbf{Task:}  #3
		\end{boxedminipage}
	\end{center}
	\vspace*{0.15cm}	}

\newcommand{\HCD}{\textsc{Hedge Cluster Deletion}\xspace}

\newcommand{\VC}{\textsc{Vertex Cover}\xspace}
\newcommand{\CD}{\textsc{Cluster Deletion}\xspace}
\newcommand{\MVC}{\textsc{Multi-Vertex Cover}\xspace}
\newcommand{\MHD}{\textsc{Min Horn Deletion}\xspace}
\newcommand{\MinO}{\textsc{MinOnes}}
\newcommand{\PropSat}[1]{\textsc{Propagational}-#1 \textsc{Satisfiability}\xspace}

\title{Algorithms and Complexity of Hedge Cluster Deletion Problems}

\author{
Athanasios L. Konstantinidis\thanks{Department of Mathematics, University of Ioannina, Greece. \texttt{a.konstantinidis@uoi.gr}.
}
\and
Charis Papadopoulos\thanks{Department of Mathematics, University of Ioannina, Greece. \texttt{charis@uoi.gr}}
\and
Georgios Velissaris\thanks{Department of Mathematics, University of Ioannina, Greece. \texttt{g.velissaris@uoi.gr}}
}
\date{}

\begin{document}

\maketitle

\begin{abstract}
A hedge graph is a graph whose edge set has been partitioned into groups called hedges. Here we consider a generalization of the well-known \textsc{Cluster Deletion} problem, named \textsc{Hedge Cluster Deletion}. The task is to compute the minimum number of hedges of a hedge graph so that their removal results in a graph that is isomorphic to a disjoint union of cliques. 
From the complexity point of view, unlike the classical \textsc{Cluster Deletion} problem, we prove that \HCD is NP-complete on non-trivial graphs which contain all graphs having a large size of vertex-disjoint 3-vertex-paths as subgraphs. To complement this, we show that for graphs that contain bounded size of vertex-disjoint 3-vertex-paths as subgraphs, \HCD can be solved in polynomial time. 

Regarding its approximability, we prove that the problem is tightly connected to the related complexity of the \textsc{Min Horn Deletion} problem, a well-known boolean CSP problem. Our connection shows that it is NP-hard to approximate \HCD within factor $2^{O(\log^{1-\epsilon} r)}$ for any $\epsilon >0$, where $r$ is the number of hedges in a given hedge graph. 
We also address the parameterized complexity of the problem. 
While \HCD is fixed-parameter tractable with respect to the solution size (i.e., the number of removal hedges), 
we prove that it does not admit a polynomial kernel, unless 
NP $\subseteq$ coNP/poly. 

Based on its classified (in)approximability and the difficulty imposed by the structure of almost all non-trivial graphs, we consider the hedge underlying structure. We give a polynomial-time algorithm with constant approximation ratio for \HCD whenever each triangle of the input graph is covered by at most two hedges. On the way to this result, an interesting ingredient that we solved efficiently is a variant of the \textsc{Vertex Cover} problem in which apart from the desired vertex set that covers the edge set, a given set of vertex-constraints should also be included in the solution. Moreover, as a possible workaround for the existence of efficient exact algorithms, we propose the hedge intersection graph which is the intersection graph spanned by the hedges. Towards this direction, we give a polynomial-time algorithm for \HCD whenever the hedge intersection graph is acyclic. 
\end{abstract}

\clearpage

\section{Introduction}
Edge deletion problems are motivated by applications in denoising data derived from imprecise experimental measurements and, therefore, have been extensively studied from classical, approximation, and parameterized complexity. 
The objective is usually concerned with the property of being $H$-free graph for some fixed graph $H$ and almost a complete classification with respect to $H$ is known \cite{KratschW13,CaiC15,BelovaB22,BliznetsCKP18,AravindSS17}. 
However, little is known when generalizing the concept from edge deletion problems to set deletion problems. 
Considering groups of edges as potential sets that can be removed is captured by the notion of hedge graphs.   
%

Hedge graphs represent situations where a set of edges might fail together due to some underlying dependency. In particular, a hedge graph is a graph whose edges are partitioned into groups called \emph{hedges}. 
Motivated by the fact that edge failures might appear simultaneously, hedge graphs appear naturally in a network design problem, known as \textsc{Hedge Cut}. In this problem, the task is to remove a minimum number of hedges that disconnects a given hedge graph. 
Indeed, in a series of recent papers there has been an increasing interest in identifying whether \textsc{Hedge Cut} is polynomial-time solvable or quasipolynomial-time solvable or NP-complete \cite{ChandrasekaranX18,sodaFominGKL025,sodaFominGKLS23,sodaGhaffariKP17,JaffkeLMPS23}.   

Edge failures that occur simultaneously do not affect only the graph connectivity. It appears to be of interest even when similarities between objects need to be identified \cite{HJ97,Hartigan}. 
Here we combine the effect of grouping edges with the task of partitioning the vertices into clusters in such a way that there should be many grouped edges within each cluster and relatively few grouped edges between the clusters. 
In particular, we consider the \HCD problem: given a hedge graph, the goal is to remove a minimum number of hedges that result in a disjoint union of cliques. 

We mention few representative domains in which the problem is well-suited. Under the term cluster graph, which refers to a disjoint union of cliques, one may find a variety of applications that have been extensively studied \cite{BBC04,CGW03,Schaeffer07}. 
In social networks it is natural to detect communities with strong relationships \cite{FORTUNATO201075}. 
These relationships can be of various types, for example, colleagues, neighbors, schoolmates, football-mates, etc., that may be represented by a different hedge. 
Another area of interest lies within protein interaction networks, in which the edges represent interactions occurring when two or more proteins bind together to carry out their biological function \cite{Ben-DorSY99}. Such interactions can be of different types, for example, physical association or direct interaction, and clusters should mainly contain edges with the same label as colocalization, that have the ability to represent a group of proteins that interact with each other at the same time and place \cite{ch16}. Similar concepts that capture the idea that edge failures occur in cascades can be applied whenever the characterization of each edge corresponds to network frequencies.

Here we concentrate on the complexity of the \HCD problem, in terms of polynomial-time algorithms, NP-completeness or approximation guarantees. Let us first discuss some observations. We denote by $\ell$ the number of hedges. It is clear that $\ell$ is at most the number of edges $m$ of the considered graph. In fact, if $\ell = m$, that is, all hedges contain a single edge, then \HCD is equivalent to the classical formulation known as \CD. 
On the other hand, if $\ell$ is bounded then the problem can be solved in polynomial time by enumerating all possible subsets of the hedges. Even further, it is not difficult to see that the problem is fixed
parameter tractable (FPT) parameterized by the number of removal hedges, i.e., by the solution size. 
This follows by using the main idea given in \cite{Cai96}: whenever we find a $P_3$, we branch at most on two hedges involved within the $P_3$. The branching will generate at most two such instances and the depth of the search tree is bounded by the solution size. The main difference with the algorithm given in \cite{Cai96} is that a single hedge removal is performed by multiple edge removals.

\subsection{Related work}
As already mentioned, hedge graphs were inspired by the \textsc{Hedge Cut} and closely related problems \cite{ChandrasekaranX18,sodaFominGKL025,sodaFominGKLS23,sodaGhaffariKP17,JaffkeLMPS23}. 
Interestingly, several hedge variants of the problem were discussed. Their main focus is on connectivity problems in order to understand the classical complexity of \textsc{Hedge Cut}. 
A quasipolynomial-time algorithm was proposed for detecting a cycle of minimum submodular cost on multigraphs \cite{sodaGhaffariKP17,FominGKLS24}. As the partitioning function into hedges is submodular (i.e., the number of hedges covering an edge subset), the same algorithm can be applied for the \textsc{Hedge Minimum Cycle} problem in which the task is to detect a cycle that contains a minimum number of hedges \cite{sodaGhaffariKP17,FominGKLS24}. However, the hedge cut function is not necessarily submodular \cite{sodaGhaffariKP17}. A recent approach that overcomes such difficulty related to hedge cut problems uses alternative measures of connectivity and proposes a polymatroid perspective of hedge graphs \cite{abs-2510-25043}.

In the \textsc{Cluster Deletion} problem we seek to delete a minimum number of edges of a given graph such that the resulting graph is a vertex-disjoint union of cliques (cluster graph).
It is known that the problem is NP-hard on general graphs \cite{CluNP04}. Settling its complexity status has attracted several researchers. If the input graph is restricted on well-known graph classes the problem remains NP-hard \cite{BDM15,CD-cographs,GolovachHKLP20, CD-lowdegree, KonstantinidisP21}, 
although there are graph classes where the problem is polynomial solvable \cite{BonomoDNV15,BDM15, CD-cographs,GruttemeierK20, CD-lowdegree, KonstantinidisP21}. In \Cref{tab:SumComplexity}, 
we summarize the complexity of the problem on the corresponding graph classes. 
\CD has been studied extensively from the perspective of approximation algorithms. A greedy approach of selecting a maximum clique provides a $2$-approximation algorithm, though not necessarily in polynomial time \cite{DessmarkJLLP07}.
Other constant factor approximation algorithms for \textsc{Cluster Deletion} were based on rounding a linear programming (LP) relaxation \cite{CGW03,PuleoM15, Veldt22a, VeldtGW18}, with the current-best approximation factor of 2 achieved in \cite{VeldtGW18}. 

As \CD is an \textsc{$H$-free Edge Deletion} problem with $H=P_3$, we also mention complexity results regarding hardness of approximation for almost every fixed graph $H$ \cite{BliznetsCKP18}. In particular, if $H$ is a 3-connected graph with at least two non-edges, or  a cycle on
at least four vertices or a path on at least five vertices, then 
\textsc{$H$-free Edge Deletion} does not admit poly(OPT)-approximation unless P = NP \cite{BliznetsCKP18}. 
Regarding its parameterized variation, several results are known when the parameter is mainly the solution size $k$. 
The classical result of Cai \cite{Cai96} implies that \textsc{$H$-free Edge Deletion} can be solved in time $c^k \cdot n^{O(1)}$, where $c$ is a constant that depends only on the size of $H$. 
Moreover, the running time is essential optimal: 
Aravind, Sandeep, and Sivadasan \cite{AravindSS17} showed that whenever $H$ contains at least two edges, there is  no algorithm with running time $2^{o(k)} \cdot n^{O(1)}$ unless the exponential time hypothesis (ETH) fails. 
For the kernelization complexity of \textsc{$H$-free Edge Deletion}, Cai and Cai \cite{CaiC15} showed that polynomial kernels do not exist (under NP $\nsubseteq$ coNP/poly) whenever $H$ is 3-edge-connected and has at least two non-edges. 
Still, there are polynomial kernels whenever $H$ is a small path \cite{CaoC12,GuillemotHPP13}. 


\begin{table}[t!]
\begin{center}
\begin{tabular}{M{0.12\textwidth}m{0.19\textwidth}m{0.22\textwidth}M{0.36\textwidth}}
\toprule
 & Restriction (graph classes) & {\raggedright {\sc Cluster} \\ {\sc Deletion}} & {\HCD} \\\midrule
\multirow{11}{\hsize}{\centering NP-hard \\ or \\ poly-time} & Bipartite & Poly-time (trivial) & \\
& Interval & Poly-time  \cite{KonstantinidisP21} &  {\centering {NP-hard}} \\
& Cograph & Poly-time~\cite{CD-cographs} & {(\Cref{theo:largesubgraphs})}\\
& Split & Poly-time~\cite{BDM15} &  \\
\cmidrule{2-4}
& $P_5$-free & NP-hard \cite{BDM15,KonstantinidisP21} & \\
& Planar & NP-hard \cite{GolovachHKLP20}  & NP-hard \\
& $\Delta \geq 4$  & NP-hard~\cite{CD-lowdegree} &  \\
\cmidrule{2-4}
& Special bipartite$^{1}$& Poly-time (trivial) & {\centering Poly-time}\\
& Special split$^{2}$& Poly-time~\cite{BDM15} & {\centering (\Cref{theo:dichotomy})} \\\cmidrule{2-4}
& {\raggedright  Acyclic hedge \\  intersection} & Poly-time (trivial) & {\centering Poly-time  \\(\Cref{theo:acyclic})}\\
\cmidrule{1-4}
\multirow{3}{0.14\textwidth}{Approxim. guarantees} &
No restriction
    & {\raggedright Constant \\ factor \cite{VeldtGW18}} & {{\MHD-complete} (\Cref{theo:mddcompleteness})} \\\cmidrule{2-4}
 &
Bi-hedge graphs
    & \multicolumn{1}{l}{\quad \quad \, \,--} &
    {\centering Constant factor \\ (\Cref{theo:approxHCD})}\\
\cmidrule{1-4}
{Parameter. complexity}
&
Parameterized by the solution size & {\raggedright Polynomial \\ kernel \cite{CaoC12,GuillemotHPP13}} & {{No polynomial kernel} (\Cref{theo:nokernel})}
\\ 
\bottomrule
\end{tabular}
%
%
\end{center}
\caption{Overview of our results on \HCD related to the complexity of {\sc Cluster Deletion} and restricted on particular graph classes. \Cref{lem:paths} (which is also used within \Cref{theo:largesubgraphs}) shows that \HCD is NP-complete on bipartite graphs, interval graphs and cographs. 
For the class of split graphs, \Cref{theo:largesubgraphs} shows the NP-completeness on $K_n - e$ graphs. 
$^{1}$Bipartite graphs with bounded size in one partition: such graphs contain a bounded number of vertex-disjoint $P_3$'s. 
$^{2}$Split graphs with bounded clique number: every $P_3$ or $K_3$ contains a vertex from the clique, implying that any collection of vertex-disjoint three-vertex connected subgraphs has bounded size. 
Moreover, if \HCD coincides with \CD then every edge belongs to a distinct hedge and the hedge intersection graph is exactly the line graph of the original graph. 
Regarding the parameterized complexity, it is known that no subexponential-time algorithm exists for \CD unless ETH fails \cite{CD-lowdegree} and such a refutation transfers for \HCD, as well. The non-existence of a polynomial kernel is claimed under the assumption NP $\nsubseteq$ coNP/poly. 
}
\label{tab:SumComplexity}
\end{table}


Moreover, clustering problems have been studied on edge-colored graphs.
In \cite{BonchiGGTU2015} the \textsc{Chromatic Correlation Clustering} problem was introduced. The input of the problem is an edge-labeled graph and the objective is to derive a partition of the vertices such that the relations among vertices in the same cluster are as much homogeneous as possible in terms of the labels \cite{Anava2015, BonchiGGTU2015, Klodt21}.  
Another related problem is the \textsc{Max $k$-Colored Clustering} problem, where given an edge-colored graph, the aim is to find a clustering of the vertices maximizing the number of matched edges, i.e., edges having the same color \cite{AlhamdanK19, ANGEL2016, KellerhalsKKN23}. In our setting we do not require to have small number of hedges that appear in the same cluster.

\subsection{Overview of our results and techniques}
Here we provide a high-level description of our results. \Cref{tab:SumComplexity} summarizes our main results. 
For necessary terminology and formal definitions concerning specially notions and terms used on hedge graphs, we refer to \Cref{sec:prelim}. However, in the following overview it is enough to recall that the edges of a hedge graph are partitioned into hedges and the underlying graph corresponds to the simple graph obtained from a hedge graph by ignoring the hedges. 

\paragraph{Hardness results} As \HCD is a generalization of \CD, it is natural to identify classes of graphs with the same or different complexity behavior. We begin by considering the underlying graph structure of a hedge graph. Clearly, if \CD is NP-complete on a certain graph class, so is \HCD on the same class. 
Our results imply that in a wide range of non-trivial graph classes \HCD is NP-complete, whereas \CD admits polynomial-time solution, highlighting a complexity different behavior (\Cref{theo:largesubgraphs}). 
More precisely, we show that if there is a subgraph that contains a large enough number of vertex-disjoint $P_3$'s then \HCD is NP-complete. Such graphs contain long paths or trees of constant height or almost cliques (graphs obtained from a clique by removing a single edge). 
Thus, in almost every interesting graph class, \HCD remains NP-complete. 

\paragraph{Tractable results} Moreover, we reveal a connection between the subgraph of disjoint $P_3$ characterization and the complexity of the problem. 
In particular, we show that if the largest subgraph of vertex-disjoint $P_3$'s is bounded, then \HCD is polynomial-time solvable (\Cref{theo:dichotomy}). For doing so, we prove that the graph obtained after removing the particular subgraph contains only an induced matching and isolated vertices. Then, by enumerating all subsolutions on the bounded-size subgraph, we are able to bound all cluster subgraphs. 
As a representative family of graphs that are applicable for the polynomial-time algorithm, we only mention complete bipartite graphs with bounded one of the two sizes, which admit trivial solution for \CD. 

\paragraph{(In)approximability} We next consider the approximation complexity of the problem. Our initial goal is to classify the problem as APX-complete or poly-APX-complete. The existence of a poly(OPT)-approximation algorithm means that there is a polynomial-time algorithm that returns a solution of cost bounded by some polynomial function of the optimum. Even though we are not able to fully classify the approximation complexity, we show that \HCD is tightly connected to a class of equivalent poly-APX problems that are still unclassified with respect to the approximation complexity. From the given class of problems there is a representative problem related to CSP, known as \MHD: given is a boolean formula $\phi$ in CNF that contains only unary clauses and clauses with three literals out of which exactly one is negative, the task is to compute the minimum number of \emph{ones} (i.e., variables assigned to \texttt{true}) such that $\phi$ is satisfied. In their pioneered work of approximability of CSPs, Khanna, Sudan, Trevisan, and Williamson proved that 
a \MHD-complete problem is poly-APX-hard
if and only if 
each \MHD-complete problem is poly-APX-hard \cite{KhannaSTW01}. 
However, they also showed that it is NP-hard to approximate a \MHD-complete problem within factor $2^{O(\log^{1-\epsilon} n)}$ for any $\epsilon >0$, where $n$ is the number of variables in a given CNF formula. 

Here we prove that \HCD is \MHD-complete under $A$-reductions, i.e., corresponding reductions required for the \MHD-completeness. Intuitively, $A$-reductions preserve approximability problems up to a constant factor (or higher). 
For the exact meaning of related notions, we refer to \Cref{sec:inapprox}. 
We point out that showing \MHD-completeness, neither rules out the possibility of the existence of poly(OPT)-approximation nor shows that such an approximation exists. 
Despite this fact, few problems are classified as \MHD-complete. 
Interestingly, the edge-deletion problem \textsc{$(K_d-e)$-Edge Deletion} was shown to be \MHD-complete in which the goal is to delete as few edges as possible from a given graph so that the resulting graph does not contain an induced subgraph isomorphic to $K_d$ without an edge, for any $d \geq 5$ \cite{BliznetsCKP18}. Subsequently, Belova and Bliznets \cite{BelovaB22} showed that the editing variation of the latter problem still remains \MHD-complete.  

To show \MHD-completeness for the \HCD problem, 
our initial \MHD-complete problem is the \MinO$(\mathcal{F})$: given a set of boolean variables and a family of characteristic constraints $\mathcal{F}$ over the variables, find a variable assignment that satisfy all of the constraints while minimizing the number of variables set to one (i.e., variables assigned to \texttt{true}). We refer to \emph{characteristic constraints}, to mean that every constraint $f$ is constructed from $\mathcal{F}$ by applying certain rules on some tuple of the variables. 
For a suitable choice of $\mathcal{F}$ so that certain properties are satisfied, it is known that \MinO$(\mathcal{F})$ is \MHD-complete \cite{KhannaSTW01}. 

We show that there is an $A$-reduction from the \MinO$(\mathcal{F'})$ to the \HCD (\Cref{lemma:monestohcd}) and there is an $A$-reduction from the \HCD to the \MinO$(\mathcal{F''})$ (\Cref{lemma:hcdtomones}). 
The key aspect of the reductions is a careful choice for the families of constraints $\mathcal{F'}$ and $\mathcal{F''}$ that encapsulate the description of \HCD, but at the same time both $\mathcal{F'}$ and $\mathcal{F''}$ fulfill the necessary hardness conditions of the \MinO \xspace problem.   
Then, by the results of \cite{KhannaSTW01} we conclude that it is NP-hard to approximate \HCD within the claimed factor (\Cref{cor:hardfactor}). Moreover, the proved completeness (\Cref{theo:mddcompleteness}) implies that \HCD problem can be poly(OPT)-approximated if and only if every \MHD-complete problem can be poly(OPT)-approximated (\Cref{cor:hardapprox}). 

\paragraph{Incompressibility} Interesting connections arise with respect to parameterized complexity. 
Recall that \HCD is easily shown to be FPT parameterized by the solution size with running time $2^{k} n^{O(1)}$. 
In fact, the running time is essentially the best one can hope for because \CD has no subexponential-time algorithm \cite{CD-lowdegree}. 
However, polynomial kernels for \CD do exist \cite{CaoC12,GuillemotHPP13}. 
Therefore, it is natural to ask whether \HCD admits a polynomial kernel. 
We answer negative to this question under reasonable hierarchy assumptions, which highlights another complexity difference between the two problems.  

More general, there seems to be a connection between the existence of polynomial kernel and poly(OPT)-approximation for a problem, because usually one can find poly(OPT)-approximation having polynomial kernel, and vice versa. Although this fact is not formal and there are problems that act differently \cite{GiannopoulouLSS16}, it is quite possible that such connection exists for edge deletion problems: up to now, the existence of poly(OPT)-approximation algorithms are similar to results for kernelization complexity \cite{BliznetsCKP18,BelovaB22,CaiC15}. 
We establish the incompressibility of \HCD by ppt-reduction from an incompressible constraint satisfiability problem (\Cref{theo:nokernel}). 
This specific problem describes propagational behaviors on Boolean variables and clauses after certain events occur, which encompasses the situation in hedge graphs after each hedge removal. 
We do so, by representing Boolean formulas by three-vertex components that ensure the consistency of $P_3$-freeness in a similar way to the \MinO$(\mathcal{F})$ problem. 
Interestingly, we also show a connection between the \MinO$(\mathcal{F})$ problem and the incompressible constraint satisfiability problem via the \HCD problem (\Cref{theo:PROPisMHD}). 
%
%
%
%
%

\paragraph{Constant-factor approximation algorithm} Based on the classified (in)approximability and the hardness of \HCD imposed by the structure of the underlying graph, we consider structural insights for establishing approximation guarantees or polynomial-time algorithms. 
In order to achieve a constant
factor approximation, we restrict ourselves to hedge graphs in which no triangle is covered by more than two hedges. By our previous results (\Cref{theo:largesubgraphs}), we note that the problem still remains NP-complete even in this special class of hedge graphs. However, we now exploit an interesting property related to \HCD. Whenever a hedge is removed, we know how to compute the set of hedges that span new $P_3$ because of the hedge removal (\Cref{lem:polyR}). For a hedge $x$, we denote by $R(x)$ all hedges that inherit such property. We show that not only we can compute $R(x)$ in polynomial time, but there is a structural property behind these sets: (P1) for any $y \in R(x)$, we have $R(y) \subseteq R(x)$ (\Cref{lem:domproperties}). 

The previous property allows us to consider an auxiliary graph with its vertex set representing the hedges of the original hedge graph. Two vertices (that correspond to two hedges) are adjacent in the auxiliary graph if there is a $P_3$ spanned by edges of the two hedges. Each vertex $x$ also carries the set $R(x)$. Now our task is to find a vertex cover of the auxiliary graph that also takes into account the set $R(x)$ for each vertex $x$ belonging to the cover (\Cref{lem:HCDtoMVC}). Though this problem sounds like a natural variant of the \VC problem, we are not aware if this variant of \VC can be approximated within a constant factor. However, we take advantage that in our problem, the considered sets $R(x)$ satisfy property (P1). 

Our next step is to show that the problem on the auxiliary graph is equivalent to a purely \VC problem having no constraints on the vertices. If between an edge $\{x,y\}$ of the auxiliary graph the sets $R(x)$ and $R(y)$ are disjoint, then we can perform edge additions that preserve the constraint that $R(x)$ or $R(y)$ is included in the solution. Fortunately, we show that this is indeed the case: consider a hedge $z$ such that $z \in R(x) \cap R(y)$. Since the edge $\{x,y\}$ must be covered by $x$ or $y$, no matter which hedge is chosen in the solution, $z$ belongs to the solution as well (\Cref{obs:capL}). Having this into account, we include all common hedges in a solution and then perform the following edge additions between the endpoints $x$ and $y$. We add all necessary edges so that the graph induced by $R(x) \cup R(y)$ (containing $x$ and $y$) is complete bipartite. Then we are in position to claim that in the resulting graph our goal is to compute a \emph{minimal} \VC (\Cref{lem:MVCtoVC}). As \VC is known to admit an $\alpha$-approximation algorithm with $\alpha \leq 2$ \cite{Karakostas09,vazirani2001}, we conclude with a polynomial-time $2$-approximation algorithm for \HCD (\Cref{theo:approxHCD}).


\paragraph{Hedge intersection graph} In order to understand further implications that might lead to the tractability of \HCD,  
we consider the natural underlying structure formed by the hedges. Our goal is to exploit properties of the hedge intersection graph that are sufficient to deduce a polynomial-time algorithm for \HCD. Here we take a first step towards this direction by developing a polynomial-time algorithm for \HCD whenever the hedge intersection graph is acyclic. 

Before reaching the details of our algorithm, let us mention some non-trivial classes of hedge graphs with unbounded number of hedges that admit acyclic hedge intersection graph. As a first example, consider any hedge graph with the edges in each biconnected component belonging to at most two hedges and all other edges belong to a distinct hedge. For dense graphs, consider a (subgraph of a) clique which is formed by a matching of unbounded size where each edge of the matching belongs to a distinct hedge and all other edges of the clique belong to the same hedge. Thus, there are some interesting classes of hedge graphs that admit an acyclic hedge intersection graph.  


The hedge intersection graph $\mathcal{F}$ of a hedge graph is defined as follows: its vertex set is the set of hedges and two hedges in $\mathcal{F}$ are adjacent if the corresponding hedges contain a common vertex in the original hedge graph. We note that since we consider $\mathcal{F}$ to be acyclic, at most two hedges cover any triangle of the hedge graph. We also note that several properties explained earlier hold in this situation as well. In particular, the sets $R(x)$ can be computed efficiently and property (P1) is fulfilled. However, here we have to be careful because between hedges in $\mathcal{F}$ there are edges caused by triangles of the hedge graph. 

It is not difficult to see that between two hedges in $\mathcal{F}$ there are two types of edges: (i) edges formed by hedges that cover only triangles or (ii) edges formed by hedges that cover only $P_3$. Of course, there are edges that belong to both types. In fact, for the latter type of edges we can deduce which of the two endpoints (hedge of the original hedge graph) should belong to any solution (\Cref{lem:F1graph}). This is based on the fact that in any triangle one of the two hedges dominates the other, meaning that the removal of one of the two hedges always cause a $P_3$. 

By removing such vertices from $\mathcal{F}$, the edges are partitioned into the two types. We now partition the vertices of $\mathcal{F}$ into the connected components spanned only by edges of type (i). Between such components we show that there is at most one edge of type (ii), because $\mathcal{F}$ is acyclic (\Cref{lem:onlyoneedge}). We then construct a graph, that we denote by $F_2$, in which we can compute a solution in polynomial time: for each type (i) edge $\{x,y\}$ of $\mathcal{F}$, we add all necessary edges in $F_2$ to make $R(x) \cup R(y)$ complete bipartite. Then we show that it is enough to compute a minimum vertex cover of $F_2$ (\Cref{lem:F2vertexcover}). In order to achieve this in polynomial time, we deduce that the constructed graph $F_2$ is bipartite (\Cref{lem:F2bipartite}). Based on the fact that a minimum vertex cover can be computed in polynomial time in bipartite graphs \cite{matching80,flow13}, we thus conclude the claimed polynomial-time algorithm (\Cref{theo:acyclic}). 


\section{Preliminaries}\label{sec:prelim}
All graphs considered here are simple and undirected. 
%
We use standard graph-theoretic terminology and refer to the textbook \cite{bookDiestel}. 

\paragraph{Basic definitions}
A graph is denoted by $G=(V,E)$ with vertex set $V$ and edge set $E$. We use the convention that $n=\lvert V \rvert$ and $m=\lvert E \rvert$.
The {\it neighborhood} of a vertex~$v$ of $G$ is $N(v)=\{x \mid vx \in E\}$ and the {\it closed neighborhood} of $v$ is $N[v] = N(v) \cup \{v\}$.
For $S \subseteq V$, $N(S)=\bigcup_{v \in S} N(v) \setminus S$ and $N[S] = N(S) \cup S$. 
A graph~$G'$ is a {\it subgraph} of $G$ if $V(G')\subseteq V(G)$ and $E(G')\subseteq E(G)$. 
For $X\subseteq V(G)$, the subgraph of $G$ {\it induced} by $X$, $G[X]$, has vertex set~$X$, and for each vertex pair~$u, v$ from $X$, $uv$ is an edge of $G[X]$ if and only if $u\not= v$ and $uv$ is an edge of $G$.
For $R\subseteq E(G)$, $G\setminus R$ denotes the graph~$(V(G), E(G)\setminus R)$, that is a subgraph of $G$ and for $S \subseteq V(G)$, $G - S$ denotes the graph~$G[V(G)-S]$, that is an induced subgraph of $G$.

For two disjoint sets of vertices $A$ and $B$, we write $E(A,B)$ to denote all edges that have one endpoint in $A$ and the other endpoint in $B$: \mbox{$E(A,B)=\{\{a,b\}: a\in A \text{ and } b\in B\}$}. A \emph{matching} in $G$ is a set of edges having no common endpoint.

A {\it clique} of $G$ is a set of pairwise adjacent vertices of $G$, and a {\it maximal clique} of $G$ is a clique of $G$ that is not properly contained in any clique of $G$.
An {\it independent set} of $G$ is a set of pairwise non-adjacent vertices of $G$.
The complete graph on $n$ vertices is denoted by $K_n$ and the chordless path on three vertices is denoted by $P_3$. If the vertices $a,b,c$ induce a $P_3$ with $\{a,c\} \notin E(G)$, we write $P_3=(a,b,c)$. 
Since we mainly deal with $P_3$ and $K_3$, we also write $K_3=(a,b,c)$ to denote that the vertices $a,b,c$ induce a $K_3$.

A graph is {\it connected} if there is a path between any pair of vertices and a {\it connected component} of $G$ is a maximal connected subgraph of $G$.
A {\it cluster graph} is a graph in which every connected component is a clique. A {\it maximal cluster subgraph} of $G$ is a maximal spanning subgraph of $G$ that is isomorphic to a cluster graph.  In the \CD problem we seek to delete
a minimum number of edges of a given graph to result in a cluster graph. 


A set of vertices $S$ is called \emph{vertex cover} of $G$ if every edge of $G$ has at least one endpoint in $S$. The \VC problem seeks for the smallest vertex cover set. 

\paragraph{Approximation algorithms} 
Let $f$ be a fixed nondecreasing function on positive integers. 
An $f(\textrm{OPT})$-approximation algorithm for a minimization problem $\Pi$ is an algorithm that finds a solution of size at most $f(\textrm{OPT})\cdot \textrm{OPT}$, where $\textrm{OPT}$ is the size of an optimal solution for a given instance of $\Pi$. We say that $\Pi$ \emph{admits $f(\textrm{OPT})$-approximation} if there exists an $f(\textrm{OPT})$-approximation algorithm for $\Pi$ that runs in polynomial time. 

\paragraph{Kernelization lower bounds} 
Here we include definitions from the parameterized complexity theory and kernelization; for further details we refer to \cite{CyganFKLMPPS15,DowneyF13}.
A problem with input size $n$ and parameter $k$ is \emph{fixed parameter tractable} (FPT), if it can be solved in time $f(k) \cdot n^{O(1)}$ for some computable function $f$.
Respectively, the complexity class FPT is composed by all fixed parameter tractable problems. 
A kernelization for a parameterized decision problem $Q$ is a
polynomial-time preprocessing algorithm that reduces any parameterized instance $(x, k)$ into an instance $(x',k')$
whose size is bounded by a function $g(k)$ and which has the same
answer for $Q$. 
The function $g$ defines the size of the kernel and the \emph{kernel has polynomial size} if $g$ is polynomial on $k$. 

A \emph{polynomial compression} of a parameterized problem $P$ into a (nonparameterized) problem $Q$ is a polynomial algorithm that takes as an input an instance $(x,k)$ of $P$ and returns an instance $x'$ of $Q$ such that (i) $(x, k)$ is a yes-instance of $P$ if and only if $x'$ is a yes-instance of $Q$ and (ii) the size of $x'$ is bounded by $p(k)$ for a polynomial $p$.
Clearly, the existence of a polynomial kernel implies that the problem admit a polynomial compression but not the other way around. 

It is well-known  that every decidable parameterized problem is FPT if and only if it admits a kernel, but it is unlikely that every problem in FPT has a polynomial kernel or polynomial compression.  
In particular, the standard \emph{composition} and \emph{cross-composition} techniques~\cite{BodlaenderDFH09,BodlaenderJK14}
allow to show that certain problems have no polynomial compressions, unless  NP $\subseteq$ coNP/poly. We will apply the following reduction, known as \emph{polynomial parameter transformation} (\emph{ppt-reduction} in
short), which plays a central role in demonstrating incompressibility (and, hence, no polynomial kernel) \cite{BodlaenderJK14}. 
A \emph{ppt-reduction} from a parameterized problem $P$ to another
parameterized problem $P'$ is an algorithm that, for input $(I, k) \in P$, takes time polynomial in $|I| + k$ and outputs an instance $(I', k') \in P'$, such that 
\begin{itemize}
\item[(i)] $(I, k)$ is a yes-instance of $P$ if and only if $(I', k')$ is a yes-instance of $P'$,
\item[(ii)] parameter $k'$ is bounded by a polynomial of $k$.
\end{itemize}
It is known that if there is a ppt-reduction from a parameterized problem $P$ to
another parameterized problem $P'$, then $P'$ admits no polynomial compression (hence
no polynomial kernel) whenever $P$ admits no polynomial compression \cite{BodlaenderJK14}.

\paragraph{Hedge graphs}
A hedge graph ${H}=(V,\mathcal{E})$ consists of a set of vertices $V$ and a set of hedges $\mathcal{E}=\{E_1, \ldots, E_{\ell}\}$. Each hedge $E_i$ of $\mathcal{E}$ is a collection of subsets of size exactly two of $V$ that are called the \emph{components} of $E_i$ (in view of graph terminology, each component corresponds to an edge). We require the components of any two hedges to be distinct (in view of graph terminology, the edges of a simple graph are partitioned into hedges). Notice that a vertex may appear in several components of the same hedge, as well as in components of different hedges. Translating the components as edges of a graph $H$, we reach the equivalent definition of a hedge graph: the edges $E(H)$ have been partitioned into hedges $E_1, \ldots, E_{\ell}$\footnote{The typical definition of hedge graphs (see e.g., \cite{sodaFominGKL025,sodaFominGKLS23,sodaGhaffariKP17}) does not require any component of a hedge to have size exactly two and two hedges may share the same component, implying parallel edges in the underlying graph. However, in our problem, components of size different than two are irrelevant and furthermore we ignore multiple edges. Nevertheless, the equivalent definition through the edge partitioning of a simple graph still applies for all considered problems.}. 
The \emph{underlying graph} of a hedge graph ${H}$ is the simple graph obtained from $H$ by simply ignoring the hedges and its edge set is exactly the union of the edges from all hedges. 

We consider a hedge graph $H$ as a simple graph (which corresponds to the underlying graph of ${H}$) in which the edges of $E(H)$ have been partitioned into $\ell$ hedges $E_1, \ldots, E_{\ell}$, denoted by $H=(V,\mathcal{E})$, where $\mathcal{E}(H)=\{E_1, \ldots, E_{\ell}\}$. Thus, each edge of $H$ belongs to exactly one hedge $E_i$. The removal of a hedge $E_i$ results in a hedge graph that we denote by $H\setminus E_i$ which is obtained from the graph $H$ by removing all edges of $E_i$ and consists of the hedge set $(E_1, \ldots, E_{\ell}) \setminus E_i$. 
%

\begin{definition}[hedge-subgraph]\label{def:hedgesubgraph}
\normalfont
Given a hedge graph $H=(V,\mathcal{E})$, a \emph{hedge-subgraph} $H'$ of $H$ consists of the same vertex set, i.e., $V(H')=V(H)$, and contains some of its hedges, i.e., $\mathcal{E}(H') \subseteq \mathcal{E}(H)$. 
Equivalently, for a subset of hedges $U \subseteq \{E_1, \ldots, E_{\ell}\}$ the hedge graph $H \setminus U$ denotes the hedge-subgraph of $H$ having hedge set $(E_1, \ldots, E_{\ell}) \setminus U$ with its underlying graph~$(V(H), E(H)\setminus E(U))$ where $E(U)$ corresponds to the set of edges of $H$ that belong to the hedges of $U$. The size $|U|$ of $U$ is the total number of hedges in $U$. 
\end{definition}

\begin{definition}[cluster hedge-subgraph]\label{def:clusterhedgesubgraph}
\normalfont
A {\it cluster hedge-subgraph} $H'$ of $H$ is a hedge-subgraph of $H$ with its underlying graph being isomorphic to a cluster graph.     
\end{definition}

We next formalize \HCD, as a decision problem. 
\vspace*{-0.01in}
\pbDef{\HCD}
	{A hedge graph $H=(V,\mathcal{E})$ and a non-negative integer $k$.}
	{Decide whether there is a subset of hedges $U \subseteq \mathcal{E}(H)$ such that $H \setminus U$ is a cluster hedge-subgraph of $H$ and $\lvert U \rvert \leq k$.}

In the optimization setting, the task of \HCD is to turn the
given hedge graph $H$ into a cluster hedge graph by deleting the minimum number of hedges. Let us discuss the complexity of \HCD with respect to the number of hedges $\ell$. It is not difficult to see that recognizing whether a given graph is a cluster graph can be done in linear time.
If $\ell = |E(H)|$, that is, all hedges of $H$ contain a single edge of $H$, then \HCD is equivalent to \CD. 
On the other hand, one can construct all hedge-subgraphs of $H$ in $2^{\ell} n^{O(1)}$ time, which means that if the number of hedges is bounded then \HCD is solved in polynomial time. In what follows, we assume that the number of hedges is not bounded.  

A natural difference that appears in the hedge variation of the problem, comes from the fact that maximal cluster subgraphs are enough to consider for the \CD problem. However, in \HCD one has to explore all cluster hedge-subgraphs rather than the maximal ones. Given a subgraph $H'$ with edge set $E(H') \subseteq E(H)$, we can identify if $H'$ is a valid solution for \HCD, i.e., whether $H'$ is a cluster hedge-subgraph of $H$. In fact, the following algorithm implies that if the number of cluster subgraphs of $H$ is polynomial then \HCD can be solved in polynomial time: for each cluster subgraph of the underlying graph, we realize if it is indeed a cluster hedge-subgraph of the given hedge graph and keep the valid solution with a maximum number of hedges. 

\begin{observation}\label{obs:hedgesubgraphs}
Given a subgraph $H'$ of a hedge graph $H$ with $E(H') \subseteq E(H)$, there is an \mbox{$O(\ell m+n)$-time} algorithm that decides whether $H'$ is a cluster hedge-subgraph of $H$, where $\ell$ and $m$ are the number of hedges and edges, respectively, of $H$. 
\end{observation}
\begin{proof}
We begin by checking if $H'$ is a cluster graph in $O(m+n)$ time. Then we search for any non-valid hedge $E_i \in \mathcal{E}(H)$ with the following properties: $E_i \cap E(H') \neq \emptyset$ and $E_i \cap (E(H)\setminus E(H')) \neq \emptyset$. If such a hedge exists, then $H'$ is not a hedge subgraph of $H$, since there is an edge in $E(H')$ that does not belong to any hedge that is fully contained in $H'$. Thus, if all hedges of $H$ are valid in $H'$, then $H'$ is indeed a hedge-subgraph of $H$, which means that $H'$ is a cluster hedge-subgraph of $H$. 
\end{proof}

Since cluster graphs are exactly the graphs that do not contain any $P_3$, we distinguish such paths with respect to the hedges of $H$. We say that a hedge $E_i$ spans an \emph{internal $P_3$} if there is a $P_3$ in $H$ having both of its edges in $E_i$. Otherwise, we refer to the induced $P_3$ of $H$ as a simple $P_3$, which means that two hedges are involved in the $P_3$ of $H$. By definition, we have the following.   

\begin{observation}\label{obs:internalP3}
Let $U$ be any solution for \HCD and let $E_i$ be a hedge of $H$ that spans an internal $P_3$. Then, $E_i \subseteq U$. 
\end{observation}

In order to solve \HCD, it is not difficult to see that the important subgraphs of a hedge graph $H$ are exactly the $P_3$'s and $K_3$'s that are contained in the underlying graph of $H$. 
We will rely on this fact several times. 
Thus we state it formally as follows. 

\begin{lemma}\label{lem:allPandK}
Let $H$ be a hedge graph on $n$ vertices. There is an $O(n^3)$-time algorithm that constructs another hedge graph $H'$ with $O(n^3)$ vertices having the following properties: 
\begin{itemize}
    \item[(i)] every connected component of $H'$ consists of exactly three vertices (so that it is either a $P_3$ or a $K_3$), and 
    \item[(ii)] any solution for \HCD on $H$ is a solution on $H'$ and vice versa.
\end{itemize}
\end{lemma}
\begin{proof}
We will use the following observation. 
Let $R$ be a set of edges and let $H \setminus R$ be the subgraph of $H$. 
The vertices of a $P_3$ in $H \setminus R$ induce either a $P_3$ or a $K_3$ in $H$. 
Based on this observation, let us assume that $H$ does not contain any connected component with at most two vertices. If such components exist, then we obtain a solution by ignoring all corresponding vertices, since they are not involved in any $P_3$ or $K_3$ in $H$. Note also that we ignore hedges that contain only edges of connected components with exactly two vertices. Thus we henceforth assume that each connected component of $H$ has at least three vertices. 

We enumerate all three-vertex subgraphs of $H$ that are isomorphic to $P_3$ or $K_3$. 
Let $\mathcal{P}_3$ be the set of all $P_3$'s and let $\mathcal{K}_3$ be the set of all $K_3$'s. 
For every element of $\mathcal{P}_3 \cup \mathcal{K}_3$, we add a connected component on three vertices in $H'$ that is isomorphic to $P_3$ or $K_3$, respectively. 
Observe that an edge $\{x,y\}$ of $H$ may appear as $\{x',y'\}$ in several connected components in $H'$, and an edge $\{x',y'\}$ of $H'$ corresponds to exactly one edge $\{x,y\}$ in $H$. Notice that by the previous assumption on $H$, any edge of $H$ belongs to some element of $\mathcal{P}_3 \cup \mathcal{K}_3$. 
Moreover, we keep exactly the same set of hedges between $H$ and $H'$: each edge $\{x',y'\}$ of $H'$ belongs to the hedge of $H$ that contains the corresponding edge $\{x,y\}$. 
Thus, $H'$ contains exactly $k$ vertices and at most $k$ edges, where $k=3 (|\mathcal{P}_3|+|\mathcal{K}_3|)$, and property (i) is satisfied by construction. 

To complete the justification for property (ii), notice the following. 
The sets of all $P_3$'s and $K_3$'s coincide in both graphs $H$ and $H'$. 
Thus $\mathcal{P}_3 \cup \mathcal{K}_3$ contains the necessary structures to obtain a valid solution for \HCD, since the set of hedges are exactly the same on $H$ and $H'$. 
Furthermore, any edge removal $\{x,y\}$ on $H$ results in removing all corresponding edges of the form $\{x',y'\}$ on $H'$ because they belong to the same hedge of $\{x,y\}$. 
Similarly, removing a hedge of $H'$ corresponds to removing all occurrences $\{x',y'\}$ of an edge $\{x,y\}$ in $H$, as they belong to the same hedge by construction. 
Therefore, the solutions on $H$ and $H'$ are equivalent. 
\end{proof}

Regarding the algorithm of \Cref{lem:allPandK}, we note that there is an enumeration algorithm that runs in $O(m+|\mathcal{P}_3|+|\mathcal{K}_3)|$ time, where $|\mathcal{P}_3|$ and $|\mathcal{K}_3|$ are the number of $P_3$'s and $K_3$'s, respectively \cite{HoangKSS13}. Moreover, it is known that $|\mathcal{K}_3| = O(m^{3/2})$ by the results in \cite{ItaiR78}. Still, our bound given in \Cref{lem:allPandK} is enough for the considered problem.

\section{Hardness result implied by the underlying structure}\label{sec:hardness}
Here we explore the complexity of \HCD. Clearly, if \CD is NP-complete on a certain graph class then \HCD remains NP-complete on the same class of graphs. Interestingly, we show that in a wide range of graph classes for which \CD is solved in polynomial time, \HCD is NP-complete. 


For $\delta\geq 0$, we denote by $\delta P_3$ the graph obtained from the union of $\delta$ vertex-disjoint $P_3$'s.

\begin{lemma}\label{lem:paths}
\HCD is NP-complete on hedge graphs with underlying graph isomorphic to $\delta P_3$. 
\end{lemma}
\begin{proof} 
Due to \Cref{obs:hedgesubgraphs} the problem is in NP. We provide a polynomial reduction from the NP-complete \VC problem. Without loss of generality, we assume that the input graph $G$ for \VC does not contain any isolated vertex.  
Given a graph $G = (V, E)$ on $n$ vertices and $m$ edges for the \VC problem, we construct a hedge graph $H = \delta P_3$ with $|V(H)| = 3m$, $|\mathcal{E}(H)| = n$ and $2m$ edges, as follows:
\begin{itemize}
    \item For each edge $e=\{x, y\} \in E(G)$, we add three vertices $e_{xy}, e_{x} \text{ and } e_{y}$ to $V(H)$. Moreover, we add the edges $\{e_{xy}, e_{x}\}$ 
and $\{e_{xy}, e_{y}\}$. 
    \item For every vertex $z$ of $G$, we add the hedge $E_z \in \mathcal{E}(H)$ where ${E}_z$ contains all the edges of the form $\{e_{uz}, e_{z}\} \in E(H)$. 
\end{itemize}
For an illustration, we refer to \Cref{fig:Hardness Cluster Deletion}. Observe that every connected component of $H$ is isomorphic to a $P_3$ and, thus, $H$ is indeed a $\delta P_3$ graph with $\delta = m$. Moreover, notice that there is a one-to-one correspondence between the edges of $G$ and the connected components of $H$. 

We claim that $G$ has a vertex cover $S$ with $|S| \leq k$ if and only if 
there is a set of $U$ hedges in $H$ with $|U|\leq k$ such that $H \setminus U$ is a cluster hedge-subgraph of $H$. 

%

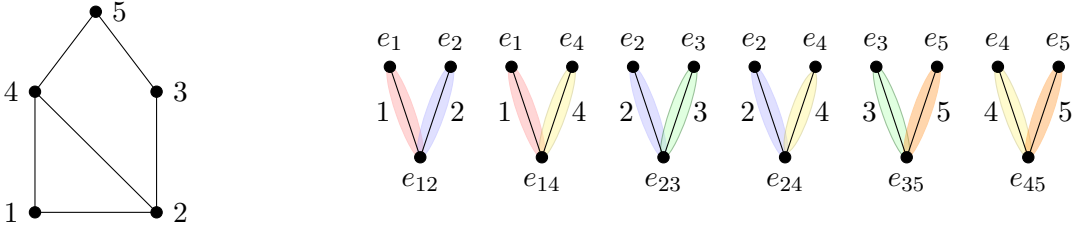
\begin{figure}[t]
\centering
\begin{tikzpicture}
\node[draw, circle, fill=black, inner sep=1.5pt, label=left:$1$] at (2.13,23.77) {};
\node[draw, circle, fill=black, inner sep=1.5pt, label=right:$2$] at (3.73,23.77) {};
\node[draw, circle, fill=black, inner sep=1.5pt, label=left:$4$] at (2.13,25.37) {};
\node[draw, circle, fill=black, inner sep=1.5pt, label=right:$3$] at (3.73,25.37) {};
\node[draw, circle, fill=black, inner sep=1.5pt, label=right:$5$] at (2.93,26.43) {};
\draw (2.13,23.77) -- (3.73,23.77);
\draw (3.73,23.77) -- (3.73,25.37);
\draw (2.13,23.77) -- (2.13,25.37);
\draw (2.13,25.37) -- (2.93,26.43);
\draw (3.73,25.37) -- (2.93,26.43);
\draw (2.13,25.37) -- (3.73,23.77);

\begin{scope}[shift={(4.8,-15.5)}]

\draw[red!50, fill=red!40, opacity=0.4, rotate around={-71:(2.2,40.6)}] 
  (2.2,40.6) ellipse (0.65cm and 0.12cm);
\draw[blue!40, fill=blue!30, opacity=0.4, rotate around={71:(2.6,40.6)}] 
  (2.6,40.6) ellipse (0.65cm and 0.12cm);
\draw[red!50, fill=red!40, opacity=0.4, rotate around={-71:(3.8,40.6)}] 
  (3.8,40.6) ellipse (0.65cm and 0.12cm);
\draw[yellow!70!black, fill=yellow!60, opacity=0.4, rotate around={71:(4.2,40.6)}] 
  (4.2,40.6) ellipse (0.65cm and 0.12cm);
\draw[blue!40, fill=blue!30, opacity=0.4, rotate around={-71:(5.4,40.6)}] 
  (5.4,40.6) ellipse (0.65cm and 0.12cm);
\draw[green!40!black, fill=green!30, opacity=0.4, rotate around={71:(5.8,40.6)}] 
  (5.8,40.6) ellipse (0.65cm and 0.12cm);
\draw[blue!40, fill=blue!30, opacity=0.4, rotate around={-71:(7.0,40.6)}] 
  (7.0,40.6) ellipse (0.65cm and 0.12cm);
\draw[yellow!70!black, fill=yellow!60, opacity=0.4, rotate around={71:(7.4,40.6)}] 
  (7.4,40.6) ellipse (0.56cm and 0.12cm);
\draw[green!40!black, fill=green!30, opacity=0.4, rotate around={-71:(8.6,40.6)}] 
  (8.6,40.6) ellipse (0.65cm and 0.12cm);
\draw[orange!90, fill=orange!80, opacity=0.4, rotate around={71:(9.0,40.6)}] 
  (9.0,40.6) ellipse (0.65cm and 0.12cm);
\draw[yellow!70!black, fill=yellow!60, opacity=0.4, rotate around={-71:(10.2,40.6)}] 
  (10.2,40.6) ellipse (0.65cm and 0.12cm);
\draw[orange!90, fill=orange!80, opacity=0.4, rotate around={71:(10.6,40.6)}] 
  (10.6,40.6) ellipse (0.65cm and 0.12cm);


\node[draw, circle, fill=black, inner sep=1.5pt, label=above:$e_{1}$] at (2,41.20) {};
\node[draw, circle, fill=black, inner sep=1.5pt, label=above:$e_{2}$] at (2.80,41.20) {};
\node[draw, circle, fill=black, inner sep=1.5pt, label=below:$e_{12}$] at (2.40,40.00) {};
\node[draw, circle, fill=black, inner sep=1.5pt, label=above:$e_{1}$] at (3.60,41.20) {};
\node[draw, circle, fill=black, inner sep=1.5pt, label=above:$e_{4}$] at (4.40,41.20) {};
\node[draw, circle, fill=black, inner sep=1.5pt, label=below:$e_{14}$] at (4.00,40.00) {};
\node[draw, circle, fill=black, inner sep=1.5pt, label=above:$e_{2}$] at (5.20,41.20) {};
\node[draw, circle, fill=black, inner sep=1.5pt, label=below:$e_{23}$] at (5.60,40.00) {};
\node[draw, circle, fill=black, inner sep=1.5pt, label=above:$e_{3}$] at (6.00,41.20) {};
\node[draw, circle, fill=black, inner sep=1.5pt, label=above:$e_{2}$] at (6.80,41.20) {};
\node[draw, circle, fill=black, inner sep=1.5pt, label=below:$e_{24}$] at (7.20,40.00) {};
\node[draw, circle, fill=black, inner sep=1.5pt, label=above:$e_{4}$] at (7.60,41.20) {};
\node[draw, circle, fill=black, inner sep=1.5pt, label=above:$e_{3}$] at (8.40,41.20) {};
\node[draw, circle, fill=black, inner sep=1.5pt, label=above:$e_{5}$] at (9.20,41.20) {};
\node[draw, circle, fill=black, inner sep=1.5pt, label=below:$e_{35}$] at (8.80,40.00) {};
\node[draw, circle, fill=black, inner sep=1.5pt, label=above:$e_{4}$] at (10.00,41.20) {};
\node[draw, circle, fill=black, inner sep=1.5pt, label=below:$e_{45}$] at (10.40,40.00) {};
\node[draw, circle, fill=black, inner sep=1.5pt, label=above:$e_{5}$] at (10.80,41.20) {};

\draw (2.00,41.20) -- (2.40,40.00) node [midway, xshift=-0.75em, yshift=0em] {$1$};
\draw (2.40,40.00) -- (2.80,41.20) node [midway, xshift=0.75em, yshift=0em] {$2$};
\draw (3.60,41.20) -- (4.00,40.00) node [midway, xshift=-0.75em, yshift=0em] {$1$};
\draw (4.00,40.00) -- (4.40,41.20) node [midway, xshift=0.75em, yshift=0em] {$4$};
\draw (5.20,41.20) -- (5.60,40.00) node [midway, xshift=-0.75em, yshift=0em] {$2$};
\draw (5.60,40.00) -- (6.00,41.20) node [midway, xshift=0.75em, yshift=0em] {$3$};
\draw (6.80,41.20) -- (7.20,40.00) node [midway, xshift=-0.75em, yshift=0em] {$2$};
\draw (7.20,40.00) -- (7.60,41.20) node [midway, xshift=0.75em, yshift=0em] {$4$};
\draw (8.40,41.20) -- (8.80,40.00) node [midway, xshift=-0.75em, yshift=0em] {$3$};
\draw (8.80,40.00) -- (9.20,41.20) node [midway, xshift=0.75em, yshift=0em] {$5$};
\draw (10.00,41.20) -- (10.40,40.00) node [midway, xshift=-0.75em, yshift=0em] {$4$};
\draw (10.40,40.00) -- (10.80,41.20) node [midway, xshift=0.75em, yshift=0em] {$5$};
\end{scope}
\end{tikzpicture}
\caption{Illustration of the reduction given in \Cref{lem:paths}. On the left side, we show an input graph $G$ for \VC and on the right side, we depict the corresponding hedge graph $H$, according to the construction.}
\label{fig:Hardness Cluster Deletion}
\end{figure}

Assume that $S \subseteq V(G)$ is a vertex cover of $G$. 
We show that $U = \{E_z \mid z \in S\}$ is a valid solution for \HCD on $H$. By construction, $H \setminus U$ is a hedge-subgraph of $H$. 
If $H \setminus U$ is not a cluster hedge-subgraph then there is a connected component 
$C = \{e_{x}, e_{xy}, e_{y}\}$ of $H\setminus U$ such that $E(C) \cap E(U) = \emptyset$. 
Since the two edges of $C$ belong to the hedges $E_x$ and $E_y$, we have $E_x, E_y \notin U$. Thus, $x,y \notin S$.
Then, however, there is an edge $\{x,y\}$ of $G$ which is not covered by a vertex of $S$ and we reach a contradiction. Therefore $H \setminus U$ results in a cluster hedge-subgraph and $|S|=|U|$. 

For the opposite direction, suppose that $U$ is a solution for \HCD.  
We claim that $S=\{z \mid E_z \in U\}$ is a vertex cover for $G$. 
By construction, for every edge $\{x,y\}$ of $G$ there is a vertex $e_{xy}$ in $H$ that is incident to two edges $\{e_{xy}, e_{x}\}$ and $\{e_{xy}, e_{y}\}$. The two edges belong to the hedges $E_x$ and $E_y$, respectively. Thus, for every connected component of $H$, at least one of its two hedges belongs to $U$, since every connected component is a $P_3$ in $H$. 
This means that for every edge $\{x,y\}$ of $G$ that corresponds to a connected component containing the vertex $e_{xy}$ of $H$, at least one of $E_x, E_y$ belongs to $U$, so that $x \in S$ or $y \in S$. 
Therefore $S$ is a vertex cover of $G$ and $|S|=|U|$. 
%
%
%
\end{proof}

Now we are ready to extend our result on a wide range of graph classes. 
We note that \Cref{lem:paths} alone is not enough to be applicable on larger underlying graphs (i.e., graphs containing union of $P_3$'s), since the rest of the edges need to be included in some of the hedges. 
For instance, connected graphs are not applicable by \Cref{lem:paths}, though they contain a large enough union of $P_3$'s.

Observe that for an underlying graph $G$, there are many hedge graphs (exponential in $m$) for which their underlying graph is isomorphic to $G$. We collect all such hedge graphs, as follows. 
Given a graph $G$, we denote by $\mathcal{H}(G)$ the class of hedge graphs that contains all hedge graphs for which their underlying graph is isomorphic to $G$. 

Given two graphs $F$ and $G$, an \emph{$F$-packing} in $G$ is a collection of pairwise vertex-disjoint copies of
$F$ in $G$. An $F$-packing in $G$ is called \emph{perfect} if it covers all the vertices of $G$. 
We write $F_3$ to denote the graph on three vertices that is connected, i.e., $F_3 \in \{P_3, K_3\}$. 
It should be noted that the decision problem of whether a graph $G$ has a perfect $F_3$-packing is NP-complete \cite{KirkpatrickH83}. However, there are simple classes of graphs in which a perfect $F_3$-packing can be computed efficiently, such as trees or complete graphs. 

Now we are ready to state our claimed NP-completeness result. We take advantage of a perfect $F_3$-packing of the underlying graph and spread the hedges in a particular way that exhibits the difficulty shown in \Cref{lem:paths}. 
We stress that we  do not claim NP-completeness on any non-cluster graph, but insist on non-cluster graphs for which a perfect $F_3$-packing can be computed in polynomial-time. 

\begin{theorem}\label{theo:largesubgraphs}
For any non-cluster graph $G$ with a perfect $F_3$-packing,
\HCD is NP-complete on the class of $\mathcal{H}(G)$ hedge graphs. 
%
%
\end{theorem}
\begin{proof}
We give a polynomial reduction from the \HCD problem on hedge graphs proved to be NP-complete in \Cref{lem:paths}. 
Let $H'$ be a hedge graph with its underlying graph isomorphic to $\delta' P_3$ having $\ell'$ hedges and $n'$ vertices. 
As we claim the construction for any non-cluster underlying graph, we will use $H'$ only to define the hedges of the given underlying graph $G$.
Let $C_1, \ldots, C_{\delta}$ be the collection of the three-vertex connected subgraphs in the given perfect $F_3$-packing of $G$, where $n=3\delta$. 
Since $G$ is non-cluster, there is a $P_3 = (x,y,z)$ in $G$. 
We denote by $X$ be the set of the connected subgraphs of the collection having at least one vertex from the $P_3$, that is, $X = \{C_i \mid C_i\cap\{x,y,z\} \neq \emptyset\}$. Note that $|V(X)| \in \{3,6,9\}$. 
Now we describe the hedge graph $H$ such that $\delta = \delta' + |X|$ and $n = n' + |V(X)|$. 
Let $C_j \in \{C_1, \ldots, C_{\delta}\}\setminus X$. Observe that $C_j$ contains two or three edges. 
We arbitrarily map two edges of $C_j$ into two edges of a connected component of $H'$ (recall that $H' = \delta' P_3$). According to the mapping for each $C_j$, we construct $\ell'$ hedges $E_1, \ldots, E_{\ell'}$ in $H$ that respect the corresponding hedges of $H'$. 
For any other edge of $H$ that it is not mapped, we include it into a single hedge $E_{\ell'+1}$. In particular, notice that the two edges of the $P_3 = (x,y,z)$ belong to $E_{\ell'+1}$. 
Thus $H$ contains $\ell'+1$ hedges.   
%
%
Now we claim that $(H', k)$ is a yes-instance if and only if $(H, k+1)$ is a yes-instance for \HCD. 

If $U'$ is a solution for \HCD on $H'$ then $H \setminus (U' \cup \{E_{\ell'+1}\})$ is a cluster hedge-subgraph because all vertices of $V(X) = H-H'$ are trivial clusters and $H'$ is a cluster hedge-subgraph. 
For the opposite direction, let $U$ be a solution for \HCD on $H$. By \Cref{obs:internalP3}, we have $E_{\ell'+1} \subseteq U$ due to the internal $P_3$ induced by $\{x,y,z\}$. 
Thus $U \setminus \{E_{\ell'+1}\}$ contains $k$ hedges. Since $H'$ is a hedge-subgraph of $H$, we conclude that $H' \setminus (U \setminus \{E_{\ell'+1}\})$ is a cluster hedge-subgraph.  
\end{proof}

As an application of \Cref{theo:largesubgraphs}, we require hedge graphs for which their underlying graph admits a polynomial-time algorithm to compute a perfect $F_3$-packing. 
Typical examples of graph classes that fall in the previous characterization are long paths and almost cliques (graphs isomorphic to $K_n$ without a single edge). Both examples reveal an interesting complexity difference with the classical variation of the problem, because on such graphs \CD is known to be solved in polynomial time, whereas by \Cref{theo:largesubgraphs} \HCD remains NP-complete.  

We next conclude the section with a polynomial-time algorithm with respect to the $F_3$-packing characterization. 
The size of an $F_3$-packing refers to the number of vertex-disjoint subgraphs in the collection. 
For a graph $G$, we denote by $\delta_{\max}$ the size of the largest $F_3$-packing. If $G$ admits a perfect $F_3$-packing then $\delta_{\max} = n/3$, whereas if $G$ is a tree of height one then $\delta_{\max} = 1$. 
For a hedge graph $H$, we also write $\delta_{\max}$ to refer to the actual number of its underlying graph. 
We show that if $\delta_{\max}$ is constant then \HCD admits a polynomial-time solution. 
In particular, our claimed algorithm has running time $n^{O(\delta_{\max})}$.  
To mention few classes of graphs with small $\delta_{\max}$, we note complete bipartite graphs with bounded one side or trees with bounded number of vertex-disjoint $P_3$'s. 

\begin{theorem}\label{theo:dichotomy}
\HCD is polynomial-time solvable on hedge graphs for which $\delta_{\max} = O(1)$. 
\end{theorem}
\begin{proof}
Let $H$ be a hedge graph and let $\delta=\delta_{\max}$. 
We show the claimed polynomial-time algorithm whenever $\delta$ is a fixed small number. 
First we show how to compute such a collection of vertex-disjoint connected subgraphs on $3$ vertices of $H$ of size $\delta$. 
We enumerate all subgraphs of $H$ on three vertices in $O(n^3)$ time. From all three-vertex subgraphs with at least two edges, we choose $\delta$ to check if they are vertex-disjoint. In total, we need $O(n^{3\delta})$ time to compute the largest $\delta$ and the corresponding $F_3$-packing of $H$. 

We now describe how to compute a solution for \HCD. In fact, we will show that all possible cluster subgraphs of $H$ can be enumerated in polynomial time. This leads to a polynomial-time algorithm by \Cref{obs:hedgesubgraphs}. 
Let $H_{\delta}$ be the subgraph of $H$ that is isomorphic to $\delta P_3$. 
By the $F_3$-packing in $H$, such a subgraph always exist.  
The $3\delta$ vertices of $H_{\delta}$ are denoted by $V_{\delta}$. 
We partition the vertices of $V(H)\setminus V_{\delta}$ into two sets: $M$ the vertices that form a maximum matching in $H - H_{\delta}$ and $Q$ the vertices that do not belong to $M$. 
We claim that there are no other edges in $H[M \cup Q]$ except the ones belonging to the matching. 
To see this, observe the following: 
\begin{itemize}
    \item any edge between vertices of $Q$ results in a matching of larger size.
    \item any edge between vertices of $M$ that are not matched results in a larger $\delta P_3$.
    \item any edge between a vertex of $M$ and a vertex of $Q$ results in a larger $\delta P_3$.
\end{itemize}
Thus any claimed edge reaches a contradiction to the maximum $\delta$ or to the maximum matching of the vertices of $M$. 

This means that in any cluster subgraph of $H$ at most two vertices of $M \cup Q$ belong to the same cluster. 
We enumerate all possible partitions of $V_{\delta}$ and for each such partition we add at most two vertices from $M \cup Q$ into each set of the partition. Since the number of sets of a given partition is at most $3\delta$ and $|M|+|Q|\leq n$, we have at most $O((n^{2})^{3\delta})=O(n^{6\delta})$ such choices. For each possible choice, we keep the ones that form a cluster subgraph of $H$. Then we check in polynomial time if such a cluster subgraph is a valid solution for \HCD on $H$ by \Cref{obs:hedgesubgraphs}.
Therefore, since $\delta$ is a fixed small number, all steps can be performed in polynomial time. 
\end{proof}

\section{(In)approximability and Incompressibility}
Here we show that \HCD is \MHD-complete and has no polynomial kernel (or polynomial compression), unless NP $\subseteq$ coNP/poly. 
When we consider the approximability, we naturally refer to the optimization version of the problem, whereas for the kernelization we refer to the decision version parameterized by the size of the solution (i.e., the number of hedge removals).  
In both directions we will apply boolean constraint satisfaction problems. 
Before introducing some relevant notation, let us first describe where the difficulty occurs when approaching the problem in both situations. 

It is not difficult to see that \HCD is closely related to \VC (just consider only the $P_3$'s in \Cref{lem:paths}). However, it appears that one needs to solve a vertex cover in a ``dynamic manner'': whenever a cover solution is found (for $P_3$-freeness), its removal results in new $P_3$’s that are not nicely structured in the original graph (except the fact that they appear in $K_3$’s). In some sense, the obstacles can propagate in the graph under each hedge removal. This chain of implications might result in a solution that is far from optimal and not within a reasonable guarantee. 
Thus, the main obstacle is that of introducing new induced copies of $P_3$ when we delete hedges, which causes a propagation of hedge deletions. 
This propagation phenomenon also arises in CSP problems when assigning \texttt{true} to Boolean variables in certain formulas that we describe next. 

We employ standard notation related to constraint satisfiability problems. 
A CNF formula $\phi$ is a conjuction of clauses where each clause is a disjunction of literals. 
A satisfying assignment for $\phi$ is an assignment on its Boolean variables such that each clause in $\phi$ is satisfied. 
For each Boolean variable, we interchange the value of one to represent \texttt{true} and zero to represent \texttt{false}. 
The \emph{degree of a Boolean variable} is the number of its occurrences in $\phi$ and $\phi$ is called \emph{3-regular} if all of its variables have degree 3. Moreover, the \emph{weight} of a truth assignment for $\phi$ is the number of ones in the assignment. 

For each clause in $\phi$, we will apply a Boolean constraint or function $f$ as follows. As we deal with 3CNF, that is, each clause is the disjunction of at most three literals, for $f(x,y,z)$ the conjuctive formula $\phi$ is of the form $
f(x_1,y_1,z_1) \wedge \cdots \wedge f(x_m,y_m,z_m)$,
 where each $f(x_i,y_i,z_i)$ is a clause of $\phi$. 

\subsection{(In)approximability of \HCD}\label{sec:inapprox}
In this section, we consider the approximability of \HCD. Our goal is to settle the approximation complexity of the problem with respect to poly(OPT)-approximation (i.e., an algorithm that returns a solution of cost bounded by some polynomial function of the optimum) and poly-APX-hardness. However we are able to show that \HCD is poly-APX-hard if and only if each problem of a particular class, known as \MHD-complete problems, is poly-APX-hard. 

Let us formally define the notion that preserves \MHD-completeness. 
The $A$-reduction applies on combinatorial optimization problems in which (i) their instances and solutions can be recognized in polynomial time, (ii) their solutions are polynomially bounded in the input size, and (iii) the objective function can be computed in polynomial time from an instance and a solution. 
%
\begin{definition}[$A$-reduction]\label{def:A-reducibility}
A problem $P$ \emph{$A$-reduces} to a problem $Q$ if there exist two polynomial-time computable functions $F$ and $G$ and a constant $\alpha$ such that for any instance ${I}$ of $P$ the following hold:
\begin{enumerate}[label = (\arabic*)]
    \item $F({I})$ is an instance of $Q$;
    \item $G({I}, S')$ is a feasible solution for ${I}$, where $S'$ is any feasible solution for $F({I})$;
    \item if $S$ is an $r$-approximate solution for $F({I})$, then $G({I}, S)$ is an $(\alpha \cdot r)$-approximate solution for ${I}$, for any $r \geq 1$.
\end{enumerate}
\end{definition}

Though $A$-reductions preserve approximability problems up to a constant factor (or higher), we will actually apply $A$-reductions as natural polynomial-time reductions. 
In fact we show that there is a polynomial-time reduction from a \MHD-complete problem to \HCD (and vice versa), that preserves the same size of the solutions. 
Our starting point is the \MHD-complete problem \MinO$(\mathcal{F})$: given a set of boolean variables and a family of characteristic constraints $\mathcal{F}$ over the variables, find a variable assignment that satisfy all of the constraints while minimizing the number of variables set to one (i.e., variables assigned to \texttt{true}). We refer to \emph{characteristic constraints}, to mean that every constraint $f$ is constructed from $\mathcal{F}$ by applying certain rules on some tuple of the variables. 
To be more precise, we define the necessary properties of boolean constraints. 
\begin{itemize}
    \item a constraint $f$ is \emph{0-valid} if the all-zeroes assignment satisfies it.
    \item a constraint $f$ is \emph{IHS-$B^+$} if it can be expressed using a CNF formula in which the clauses are all of one of the following types: $x_1 \wedge \cdots \wedge x_k$ for some positive integer $k \leq B$, or $\neg x_1 \wedge  x_2$, or $\neg x_1$.
    IHS-$B^-$ constraints are defined analogously, with every literal being replaced by its complement. Moreover, IHS-$B$ constraint is either IHS-B$^+$ or IHS-B$^-$. 
    \item a constraint $f$ is \emph{weakly positive} if it can be expressed using a CNF formula that has at most one negated variable in each clause.
\end{itemize}

The above constraints can be extended to a family of constraints.
For a family of constraints $\mathcal{F}$, we say that $\mathcal{F}$ is one of the above classes if all of its constraints belong to the same class.  
Observe that $A$-reductions are transitive. 
For two problems $P$ and $Q$, we say that $P$ is \emph{$Q$-complete under $A$-reductions} if there are $A$-reductions from $P$ to $Q$ and from $Q$ to $P$. 
Now we are in position to formally state the following result. 

\begin{theorem}[\cite{KhannaSTW01}]\label{theo:minones}
If a family of constraints $\mathcal{F}$ is weakly positive, but it is neither 0-valid nor IHS-$B$ for any constant $B$, then \MinO$(\mathcal{F})$ is \MHD-complete under $A$-reductions. 
Consequently, it is NP-hard to approximate \MinO$(\mathcal{F})$ within factor 
$2^{O(\log^{1-\epsilon} n)}$ for any $\epsilon >0$, where $n$ is the number of variables in the given instance. 
\end{theorem}

We turn our attention to \HCD. For a hedge graph $H$ we enumerate all $P_3$'s and $K_3$'s and apply \Cref{lem:allPandK}. 
We will build a family of constraints in a way that every boolean variable $x$ corresponds to exactly one hedge $E_x$ and a variable $x$ is set to $1$ if and only if the corresponding hedge $E_x$ is removed.   
Depending on the number of hedges that cover each $P_3$ and $K_3$, we define the following constraints: 
\begin{itemize}
    \item $f_1(x_1, x_2, x_3)$ is equal to zero if and only if exactly one of the variables $x_1,x_2,x_3$ is set to $1$. That is, $f_1(x_1, x_2, x_3) = (\lnot x_1 \lor x_2 \lor x_3) \land (x_1 \lor \lnot x_2 \lor x_3) \land (x_1 \lor x_2 \lor \lnot x_3)$. We will use $f_1$ to encapsulate a $K_3$ that is covered by exactly three hedges.
    \item $g_1(x_1,x_2)$ is equal to zero if and only if exactly $x_1$ is set to $1$. That is, $g_1(x_1,x_2)=(\lnot x_1 \lor x_2)$. We will use $g_1$ to encapsulate a $K_3$ that is covered by exactly two hedges in which the two edges of the $K_3$ belong to the hedge ${x_1}$. 
    \item $g_2(x_1,x_2)$ is equal to zero if and only if both $x_1$ and $x_2$ are set to $0$. That is, $g_2(x_1,x_2)=(x_1 \lor x_2)$. We will use $g_2$ to encapsulate a $P_3$ that is covered by two hedges. 
    \item $f_2(x)=x$. We will use $f_2$ to encapsulate a $P_3$ that is covered by exactly one hedge.     
\end{itemize}

Notice that we do not define a constraint for a triangle that is covered by exactly one hedge.  


\begin{lemma}[\cite{BliznetsCKP18}]\label{lemma:firstFproperties}
The family of constraints $\mathcal{F'}=\{f_1, f_2\}$ is weakly positive, and at the same time it is
neither 0-valid nor IHS-$B$ for any $B$. 
\end{lemma}

Therefore, \Cref{theo:minones} and \Cref{lemma:firstFproperties} imply that \MinO($\mathcal{F'}$) is \MHD-hard under A-reductions. We now show the first reduction from \MinO($\mathcal{F'}$). 

\begin{lemma}\label{lemma:monestohcd}
Let $\mathcal{F'}=\{f_1, f_2\}$. 
Then, \MinO$(\mathcal{F'})$ $A$-reduces to \HCD. 
\end{lemma}
\begin{proof}
    We give a polynomial-time computable transformation from an instance of \MinO$(\mathcal{F'})$ to an instance of \HCD.
    Given a formula $\phi$, for each variable $x$ in $\phi$ we create a hedge in $H$, denoted by $X$. 
    Now, for every constraint, we construct a specific $K_3$ or $P_3$ in $H$.
    \begin{itemize}
        \item For every constraint $f_1(x_1, x_2, x_3)$, we construct a vertex-disjoint $K_3$ in $H$, where its three edges belong to $X_1,X_2,X_3$, respectively. 
        \item For each constraint $f_2(x)$, we construct a vertex-disjoint $P_3$ whose two edges belong to hedge $X$.
    \end{itemize}
    Observe that there is a one-to-one correspondence between the hedges of $H$ and the variables of $\phi$. Moreover, every connected component is $K_3$ or $P_3$ and the number of connected components is at most the number of constraints.
    
    We claim that there is a satisfying assignment for $\phi$ that sets exactly $k$ variables to 1 if and only if there is a set of hedges $U$ in $H$ such that $H \setminus U$ is a cluster hedge-subgraph of $H$ and $|U|=k$. Notice that the constant $\alpha$ of the A-reduction given in \Cref{def:A-reducibility} is equal to 1.
    
    Assume that there exists a satisfying assignment for $\phi$ with exactly $k$ variables set to $1$. We show that it is enough to delete the corresponding hedges from $H$, that is $U = \{X \mid x=1\}$.
    We know that for each constraint of type $f_1$, the formula is evaluated to $0$ if and only if exactly one of the variables $x_1, x_2, x_3$ is set to $1$. Consequently:
    \begin{itemize}
        \item If none of the variables are set to $1$, no hedge is deleted from the corresponding $K_3$ in $H$.
        \item If exactly two variables are set to $1$, the two corresponding hedges are deleted, leaving a $K_2$ and an isolated vertex in $H$, which constitute a valid cluster.
        \item If all three variables are set to $1$, all hedges of the corresponding $K_3$ are deleted, resulting in three isolated vertices, each of which forms a valid cluster.
    \end{itemize}
    
    For each constraint of type $f_2$, the corresponding variable is always set to $1$. Consequently, the associated hedge is deleted, and each initial $P_3$ of the construction corresponds to three isolated vertices, since both of its edges belong to the same hedge, which is guaranteed to be removed.
    
    Hence, no new $P_3$ structures can emerge from any $K_3$ in $H$ after the deletions, as the presence of such a path would imply that, for some constraint of type $f_1$, exactly one variable was assigned the value $1$, which contradicts the assumption that $\phi$ is satisfied.

    For the opposite direction, assume a solution $U$ for \HCD on $H$ with exactly $k$ hedges. We claim that for every hedge of $U$, it is enough to set the corresponding variable to 1. 
Let $x$ be a variable that is set to 1 with the corresponding hedge $X \in U$.    
If the constraint is of type $f_2$, then there is a connected component in $H$ isomorphic to $P_3$ with both of its edges belonging to the hedge $X$. Thus $X \in U$ which means that $f_2$ is satisfied. 
Now consider a constraint of type $f_1$. Assume that there is a constraint $f_1$ that is not satisfied. By definition, $f_1$ is equal to zero if and only if exactly one of its three variables is set to 1. Without loss of generality, assume that $x_1 \in U$ and $x_2,x_3 \notin U$. Then, however, the corresponding connected component of $f_1$ spans a $P_3$ in $H \setminus U$ which contradicts the fact that $H \setminus U$ is a cluster hedge-subgraph. 
Therefore, the claimed assignment satisfies $\phi$, as desired. 
%
%
\end{proof}

Next we consider the completeness with respect to $A$-reductions. We extend the family $\mathcal{F'}$ to include all defined constraints. 

\begin{lemma}\label{lemma:secondFproperties}
The family of constraints $\mathcal{F''}=\{f_1, g_1, g_2, f_2\}$ is weakly positive, and at the same time it is
neither 0-valid nor IHS-$B$ for any $B$. 
\end{lemma}
\begin{proof}
Observe that $\{f_1, f_2\} = \mathcal{F'} \subset \mathcal{F''}=\{f_1, g_1, g_2, f_2\}$. 
Thus $\mathcal{F}''$ is neither 0-valid nor IHS-$B$ for any $B$ by \Cref{lemma:firstFproperties}. 
We show that $\mathcal{F}''$ is weakly positive. Each clause appearing in $g_1$ and $g_2$ contains at most one negated literal. Since $f_1, f_2$ are weakly positive by \Cref{lemma:firstFproperties}, we conclude that 
the entire family $\mathcal{F}''$ satisfies the claimed properties.
%
%
\end{proof}

Again, notice that \Cref{theo:minones} and \Cref{lemma:secondFproperties} show that \MinO($\mathcal{F''}$) is \MHD-hard under A-reductions. Next, we show a reduction in the other direction, from \HCD to \MinO$(\mathcal{F''})$. 

\begin{lemma}\label{lemma:hcdtomones}
Let $\mathcal{F''}=\{f_1, g_1, g_2, f_2\}$. Then, 
\HCD $A$-reduces to \MinO$(\mathcal{F''})$.  
\end{lemma}
\begin{proof}

Let $H$ be a hedge graph given as an instance of \HCD. By \Cref{lem:allPandK}, we assume that every connected component of $H$ contains exactly three vertices. We show how to construct an instance of \MinO$(\mathcal{F}'')$ from $H$. 
    In order to construct a formula $\phi$ we consider every $K_3$ or $P_3$ in each connected component of $H$ and add the appropriate constraint according to the configuration of their hedges:
    \begin{itemize}
        \item For every $K_3$ whose three edges belong to distinct hedges, we add the constraint $f_1(x_1, x_2, x_3)$, where the variables $x_1$, $x_2$, and $x_3$ correspond to hedges $X_1$, $X_2$, and $X_3$, respectively.
    
        \item For every $K_3$ where two edges belong to the same hedge and the remaining edge belongs to a different hedge, we add the constraint $g_1(x_1, x_2)$, where $x_2$ represents the hedge $X_2$ containing two of the edges and $x_1$ represents the hedge $X_1$ containing the remaining edge.

        \item For every $K_3$ whose three edges belong to a single hedge, we do not add any constraint and ignore such a connected component.
        
        \item For every $P_3$ whose two edges belong to different hedges, we add the constraint $g_2(x_1, x_2)$, where the variables $x_1$ and $x_2$ correspond to hedges $X_1$ and $X_2$, respectively.
    
        \item For every $P_3$ whose two edges both belong to the same hedge, we add the constraint $f_2(x)$, where $x$ represents the hedge $X$.
    \end{itemize}
This completes the construction of the instance $\phi$ for \MinO$(\mathcal{F}'')$ (see also \Cref{fig:MinOnes-and-Propagation}). It is not difficult to see that all steps can be performed in polynomial time. 

Let $\ell$ be the number of hedges in $H$. We claim the following: for any $I \subseteq \{1, \ldots, \ell\}$, 
$H \setminus U$ where $U = \{X_i \mid i \in I\}$ is a cluster hedge-subgraph if and only if the assignment $\{x_i=1 \mid i \in I\}$ satisfies $\phi$. 

If $U$ is a solution for \HCD on $H$ then the original $P_3$'s in $H$ are not present in $H \setminus U$. This means that the constraints $g_2$ and $f_2$ have at least one of their variables set to 1. Now consider the constraints $f_1$ and $g_1$ that correspond to connected components isomorphic to $K_3$. For a $K_3$, let $H(K_3)$ be the set of hedges covered by the $K_3$. If $|H(K_3)|=3$ then we know that one of the following cases occur in $H \setminus U$: $|H(K_3) \cap U| \in \{0,2,3\}$. This is exactly described by the constraint $f_1$, whenever $f_1$ is satisfied. Now assume that $|H(K_3)|=2$ and let $X_1,X_2$ be the hedges covered by the $K_3$, where $X_2$ covers two of the edges of the $K_3$. If $X_1 \in U$ and $X_2 \notin U$ then $H \setminus U$ is not a cluster hedge-subgraph as there is a $P_3$ spanned by edges of $X_2$. Otherwise, constraint $g_1$ is satisfied. Therefore, the claimed assignment satisfies $\phi$.

For the opposite direction, assume that the assignment $\{x_i=1 \mid i \in I\}$ satisfies $\phi$. 
We show that every connected component of $H$ is a cluster graph in $H \setminus U$. 
Constraints $g_2$ and $f_2$ guarantee that each connected component of $H$ isomorphic to $P_3$ in $H$ does not contain at least one of its two edges in $H \setminus U$. For every connected component of $H$ isomorphic to $K_3$, we apply constraints $f_1$ or $g_1$ depending on the number of hedges that covers the $K_3$. We note that if the $K_3$ is covered by exactly one hedge $X$ then regardless of whether $X \in U$ or not, the three vertices of the connected component induce a cluster graph in $H \setminus U$. If the $K_3$ is covered by two hedges $X_1,X_2$ with $X_1$ the hedge that covers the unique edge of $K_3$, then the constraint $g_2$ forces that $X_1 \in U$ and $X_2 \notin U$ never occurs. Otherwise, the $K_3$ is covered by three hedges and the constraint $f_1$ is responsible not to remove exactly one of the two edges. 
Therefore, $H \setminus U$ is a cluster hedge-subgraph of $H$, as claimed. 
\end{proof}

By combining \Cref{lemma:monestohcd,lemma:hcdtomones} and \Cref{theo:minones} we conclude with the following result.

\begin{theorem}\label{theo:mddcompleteness}
    \HCD is \MHD-complete under A-reductions. 
\end{theorem}

We also state the claimed results. In particular, \Cref{lemma:firstFproperties,lemma:monestohcd} and \Cref{theo:minones} imply \Cref{cor:hardfactor}, whereas \Cref{theo:minones,theo:mddcompleteness} and $A$-reductions imply \Cref{cor:hardapprox}. 

\begin{corollary}\label{cor:hardfactor}
It is NP-hard to approximate \HCD within factor 
$2^{O(\log^{1-\epsilon} r)}$ for any $\epsilon >0$, where $r$ is the number of hedges in the given hedge graph.
\end{corollary}

\begin{corollary}\label{cor:hardapprox}
\HCD admits an $n^{\delta_1}$-approximation algorithm for any $\delta_1 > 0$ if and only if 
each \MHD-complete problem admits an $n^{\delta_2}$-approximation algorithm for any $\delta_2 > 0$.   
\end{corollary}

\subsection{Incompressibility of \HCD}\label{sec:incompres}
Recall that with a simple brute-force algorithm \HCD can be solved in $2^{k} \cdot n^{O(1)}$ time. 
Here we show that \HCD does not have a polynomial kernel parameterized by $k$ unless the polynomial hierarchy collapses.

A ternary boolean function $f(x,y,z)$, where $x, y$ and $z$ are either
boolean variables or constants 0 or 1, is \emph{propagational} if it satisfies 
\begin{itemize}
    \item $f (1, 0, 0) = 0$ and
    \item $f (0, 0, 0) = f (1, 0, 1) = f (1, 1, 0) = f (1, 1, 1) = 1$.
\end{itemize}
It is not difficult to see that the following Boolean function is propagational \cite{CaiC15,KratschW13}:
$$
\texttt{Not-1-in-3}(x, y,z) = f_1(x, y, z) = (\lnot x \lor y \lor z) \land (x \lor \lnot y \lor z) \land (x \lor y \lor \lnot z).
$$
Note that $f_1$ is symmetric ternary function and the order of the three arguments in $f_1$ is not important. Within the list of arguments, we write first the constants (i.e., $f_1(0,x,y)=f_1(x,0,y)$, $f_1(0,1,x)=f_1(1,0,x)$). Also notice that $f_1 \in \mathcal{F'} \cap \mathcal{F''}$ in our previous considered \MinO \xspace problem which was used to encode every $K_3$ that contains three hedges. 

We will reduce from the \textsc{Propagational-$f$ Satisfiability} that is defined as follows: given a conjunctive formula $\phi$ of a propagational ternary function $f$ with distinct variables inside each clause of $\phi$ and an integer parameter $k \geq 0$, the task is to decide whether there exists a satisfying truth assignment for $\phi$ of weight at most $k$.
Cai and Cai \cite{CaiC15} introduced the problem and proved that it has no polynomial kernel under plausible hierarchy assumptions.    

\begin{theorem}[\cite{CaiC15}]\label{theo:propagational}
For any propagational ternary boolean function $f$, 
\textsc{Propagational-$f$ Satisfiability} 
on 3-regular conjunctive formulas admits no polynomial compression,
hence no polynomial kernel, unless NP $\subseteq$ coNP/poly. 
\end{theorem}

In what follows, we set $f_1$ as a propagational function and consider the \textsc{Propagational-$f_1$ Satisfiability} problem. For technical reasons, it is more convenient to deal with $\phi$ that does not contain a function of the form $f_1(0,0,x)$. Next we show how this assumption is achieved.

\begin{lemma}\label{lem:f1constant}
Let $(\phi,k)$ be an instance of \textsc{Propagational-$f_1$ Satisfiability}. In polynomial time we can compute a Boolean formula $\phi'$ from $\phi$ such that $f_1(0,0,x) \notin \phi'$ and 
$(\phi,k)$ is a yes-instance if and only if $(\phi',k)$ is a yes-instance. 
\end{lemma}
\begin{proof}
We construct $\phi'$ from $\phi$ by removing certain clauses. 
For each Boolean variable $x$, we check whether $f_1(0,0,x) \in \phi$. If this is not the case for every $x$, then $\phi'=\phi$ and we conclude with the instance $(\phi',k)$. 
Assume that $f_1(0,0,x) \in \phi$.  
By the definition of $f_1$, observe that $\phi$ is satisfied only if $x=0$. 
We construct $\phi'$ from $\phi$ by replacing every occurrence of $x$ with the constant $0$ in $\phi$. 
Since we have not assigned one to any variable during this construction, it is not difficult to see that $\phi$ is satisfied by at most $k$ ones if and only if $\phi'$ is satisfied by at most $k$ ones. 
Notice also that clauses in $\phi'$ that contain only constants can be safely removed, because if we encounter $f_1(0,0,1) \in \phi'$ then $(\phi,k)$ is a no-instance; otherwise we just ignore such clauses (as they are always \texttt{true}). 
By applying such a transformation for every variable, we conclude with a valid $\phi'$ such that $f_1(0,0,x) \notin \phi'$. 
\end{proof}

We next establish the claimed incompressibility of \HCD 
mainly by a ppt-reduction from \textsc{Propagational-$f_1$ Satisfiability}. 
The proof is similar to the reduction given in \Cref{lemma:monestohcd}. 

\begin{figure}[!t]
\centering
\begin{tikzpicture}[scale=1.5]

\begin{scope}[shift={(-4.5,18)}] 
\node[anchor=west] at (0,2.2){\MinO$(\mathcal{F}'')$ formula:};
\node[anchor=west] at (0,1.8) {$\phi = f_1(x_1,x_2,x_3) \land g_1(x_4,x_3) \land $};
\node[anchor=west] at (0.5,1.3) {$g_1(x_3,x_4) \land f_1(x_1,x_3,x_5) \land $} ;
\node[anchor=west] at (0.5,0.8) {$f_2(x_1) \land g_2(x_2,x_3)$} ;
\node[anchor=west] at (0,0){\textsc{Propagational-$f_1$ Satisfiability} formula:};
\node[anchor=west] at (0,-0.5) {$\phi = f_1(x_1,x_2,x_3) \land f_1(0,x_3,x_4) \land $};
\node[anchor=west] at (0.5,-1.0) {$f_1(0,x_4,x_3) \land f_1(x_1,x_3,x_5) \land $} ;
\node[anchor=west] at (0.5,-1.5) {$f_1(1,1,x_6) \land f_1(0,1,x_1) \land f_1(1,x_2,x_3)$} ;

\draw[double, double distance=1pt] (-0.1, 0.4) -- (5.5, 0.4);
\end{scope}

\begin{scope}[shift={(0,0)}]
\node[draw, circle, fill=black, inner sep=1.5pt
] at (1.60,20.00) {};
\node[draw, circle, fill=black, inner sep=1.5pt
] at (2.80,20.40) {};

\node[draw, circle, fill=black, inner sep=1.5pt
] at (2.80,19.60) {};
\node[draw, circle, fill=black, inner sep=1.5pt
] at (4.00,20.00) {};
\node[draw, circle, fill=black, inner sep=1.5pt
] at (5.20,20.40) {};

\node[draw, circle, fill=black, inner sep=1.5pt
] at (5.20,19.60) {};
\node[draw, circle, fill=black, inner sep=1.5pt
] at (1.60,18.40) {};
\node[draw, circle, fill=black, inner sep=1.5pt
] at (2.80,18.80) {};

\node[draw, circle, fill=black, inner sep=1.5pt
] at (2.80,18.00) {};
\node[draw, circle, fill=black, inner sep=1.5pt
] at (4.00,18.40) {};
\node[draw, circle, fill=black, inner sep=1.5pt
] at (5.20,18.80) {};

\node[draw, circle, fill=black, inner sep=1.5pt
] at (5.20,18.00) {};
\node[draw, circle, fill=black, inner sep=1.5pt
] at (1.60,16.80) {};
\node[draw, circle, fill=black, inner sep=1.5pt
] at (2.80,17.20) {};

\node[draw, circle, fill=black, inner sep=1.5pt
] at (2.80,16.40) {};
\node[draw, circle, fill=black, inner sep=1.5pt
] at (4.00,17.20) {};
\node[draw, circle, fill=black, inner sep=1.5pt
] at (4.80,17.20) {};

\node[draw, circle, fill=black, inner sep=1.5pt
] at (4.00,16.40) {};
\node[draw, circle, fill=black, inner sep=1.5pt
] at (4.80,16.40) {};

\node[draw, circle, fill=black, inner sep=1.5pt
] at (5.60,17.20) {};
\node[draw, circle, fill=black, inner sep=1.5pt
] at (5.60,16.40) {};

\draw (1.60,20.00) -- (2.80,20.40) node[midway, above] {$x_1$};
\draw (1.60,20.00) -- (2.80,19.60) node[midway, below] {$x_2$};
\draw (2.80,20.40) -- (2.80,19.60) node[midway, right] {$x_3$};

\draw (4.00,20.00) -- (5.20,20.40) node[midway, above] {$x_3$};
\draw (4.00,20.00) -- (5.20,19.60) node[midway, below] {$x_3$};
\draw (5.20,20.40) -- (5.20,19.60) node[midway, right] {$x_4$};

\draw (1.60,18.40) -- (2.80,18.80) node[midway, above] {$x_4$};
\draw (1.60,18.40) -- (2.80,18.00) node[midway, below] {$x_4$};
\draw (2.80,18.80) -- (2.80,18.00) node[midway, right] {$x_3$};

\draw (4.00,18.40) -- (5.20,18.80) node[midway, above] {$x_1$};
\draw (4.00,18.40) -- (5.20,18.00) node[midway, below] {$x_3$};
\draw (5.20,18.80) -- (5.20,18.00) node[midway, right] {$x_5$};

\draw (1.60,16.80) -- (2.80,17.20) node[midway, above] {$x_6$};
\draw (1.60,16.80) -- (2.80,16.40) node[midway, below] {$x_6$};
\draw (2.80,17.20) -- (2.80,16.40) node[midway, right] {$x_6$};

\draw (4.00,17.20) -- (4.80,17.20) node[midway, below] {$x_1$};
\draw (4.80,17.20) -- (5.60,17.20) node[midway, below] {$x_1$};
\draw (4.00,16.40) -- (4.80,16.40) node[midway, below] {$x_2$};
\draw (4.80,16.40) -- (5.60,16.40) node[midway, below] {$x_3$};
\end{scope}

\draw[double, double distance=1pt] (1.0, 16.0) -- (1.0, 21.0);

\end{tikzpicture}
\caption{Illustrating CNF formulas for \MinO$(\mathcal{F}'')$ and \textsc{Propagational-$f_1$ Satisfiability}, alongside the corresponding hedge graph given in \Cref{lemma:hcdtomones} and \Cref{theo:nokernel}, respectively. Note that the $K_3$ spanned solely by hedge ${x_6}$ does not correspond to any constraint in the \MinO$(\mathcal{F}'')$ formulation.
}
\label{fig:MinOnes-and-Propagation}
\end{figure}
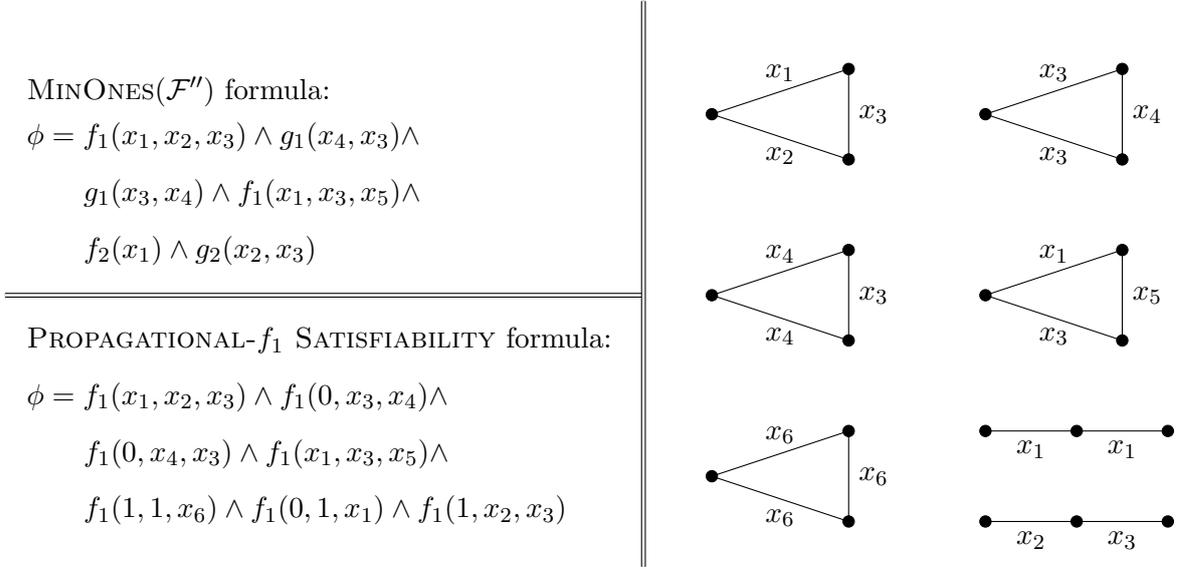

\begin{theorem}\label{theo:nokernel}
\HCD parameterized by the solution size admits no polynomial compression,
hence no polynomial kernel, unless NP $\subseteq$ coNP/poly. 
\end{theorem}
\begin{proof}
We give a ppt-reduction from \textsc{Propagational-$f_1$ Satisfiability} to the parameterized problem of \HCD. 
Let $(\phi,k)$ be an instance of \textsc{Propagational-$f_1$ Satisfiability}. 
By \Cref{lem:f1constant}, we assume that $f_1(0,0,x) \notin \phi$. 
As $f_1$ is symmetric ternary and propagational function, the following forms can only occur in $\phi$: 
$f_1(0,1,x)$, $f_1(1,1,x)$, $f_1(0,x,y)$, $f_1(1,x,y)$, $f_1(x,y,z)$,
where $x,y,z$ are Boolean variables. 
We construct a hedge graph $H$ from $\phi$ in which each hedge $E_x$ of $H$ corresponds to exactly one Boolean variable $x$ of $\phi$. For each clause in $\phi$, we construct some connected components in $H$ where every connected component has exactly three vertices in $H$ (so that every component is a $P_3$ or a $K_3$). Depending on the five formulations of $f_1$, we construct vertex-disjoint connected components as follows: 
\begin{itemize}
    \item $f_1(0,1,x)$: add a $P_3$ that is covered by exactly one hedge $E_x$.
    \item $f_1(1,1,x)$: add a $K_3$ that is covered by exactly one hedge $E_x$.
    \item $f_1(0,x,y)$: add two $K_3$'s where the edges of the two components are covered by exactly two hedges $E_x$ and $E_y$, such that the two edges of the one component belong to $E_x$ and the two edges of the other component belong to $E_y$.
    \item $f_1(1,x,y)$: add a $P_3$ that is covered by exactly two hedges $E_x,E_y$.
    \item $f_1(x,y,z)$: add a $K_3$ that is covered by exactly three hedges $E_x, E_y, E_z$. 
\end{itemize}
This concludes the construction of $H$ (see also \Cref{fig:MinOnes-and-Propagation}). It is not difficult to see that the construction can be performed in polynomial time. 
Now we claim that $(\phi,k)$ is a yes-instance of \textsc{Propagational-$f_1$ Satisfiability} if and only if
$(H,k)$ is a yes-instance of \HCD. 

Let $X'$ be a subset of the Boolean variables such that $\{x=1 \mid x\in X'\}$ satisfies $\phi$. 
We show that $H \setminus E(X')$ is a cluster hedge-subgraph where $E(X') = \{E_x \mid x \in X'\}$. 
Assume for contradiction that there is $P_3=(u,v,w)$ in $H \setminus E(X')$. This means that the hedges $E_x$ and $E_y$ of the edges $\{u,v\}$ and $\{v,w\}$, respectively, do not belong to $E(X')$, so that for the corresponding variables we have $x=y=0$ (notice that $E_x,E_y$ are not necessarily different hedges). Moreover, if $\{u,w\}$ is an edge of $H$ then we know that $E_z \notin \{E_x, E_y\}$ where $E_z$ is the hedge of the edge $\{u,w\}$ and for its corresponding variable we have $z=1$. Depending on whether $E_x = E_y$ and $\{u,w\} \in E(H)$, we conclude that a particular clause in $\phi$ is not satisfied:
\begin{itemize}
    \item $E_x = E_y$ and $\{u,w\} \notin E(H)$. Then $f_1(0,1,x)=0$, because $x=0$;
    \item $E_x = E_y$ and $\{u,w\} \in E(H)$. Then $f_1(0,x,z)=0$, because $x=0$ and $z=1$;
    \item $E_x \neq E_y$ and $\{u,w\} \notin E(H)$. Then $f_1(1,x,y) = 0$, because $x=y=0$;
    \item $E_x \neq E_y$ and $\{u,w\} \in E(H)$. Then $f_1(x,y,z)=0$, because $x=y=0$ and $z=1$. 
\end{itemize}
Therefore, in each case we conclude with a clause that is not satisfied, leading to a contradiction, which implies that there is no $P_3$ in $H \setminus E(X')$.  

For the opposite direction, let $H' = H \setminus E(X')$ be a cluster hedge-subgraph where $E(X')$ is a subset of the hedges of $H$. Then we claim that $\{x=1 \mid x\in E(X')\}$ satisfies $\phi$. 
We show that in each formulation of $f_1$ in $\phi$, $f_1(x,y,z)$ is satisfied. 
\begin{itemize}
    \item $f_1(0,1,x)$: the corresponding $P_3$ is not present in $H'$, so that $x\in E(X')$ and $x=1$ which means that $f_1(0,1,x)$ is satisfied. 
    \item $f_1(1,1,x)$: there are two cases to consider in $H'$. Either the corresponding $K_3$ remains the same ($x=0$) or the three vertices belong to singleton clusters ($x=1$). Thus, $f_1(1,1,x)$ is satisfied. 
    \item $f_1(0,x,y)$: there are two cases to consider in $H'$. Either the two corresponding components are untouched ($x=0$ and $y=0$) or all six vertices belong to singleton clusters ($x=1$ and $y=1$). Thus, $f_1(0,x,y)$ is satisfied. 
    \item $f_1(1,x,y)$: the corresponding $P_3$ is not present in $H'$, so that $x\in E(X')$ and $x=1$ or $y\in E(X')$ and $y=1$ (or both belong to $E(X')$ and $x=y=1$). Thus, $f_1(1,x,y)$ is satisfied. 
    \item $f_1(x,y,z)$: either all edges of the corresponding $K_3$ are present in $H'$ ($x=y=z=0$) or at least two of its edges are removed ($x=y=1$). Thus, $f_1(x,y,z)$ is satisfied.
\end{itemize}
Therefore, all clauses of $\phi$ are satisfied as claimed. 
\end{proof}

\section{Towards a constant-factor approximation}\label{sec:approx}
Here we consider approximating \HCD up to a constant factor, by restricting to a certain class of hedge graphs. 
In order to achieve a constant factor approximation, we require a relationship between the hedges and the triangles of the underlying graph.  
In particular, we consider hedge graphs in which every triangle of the underlying graph is covered by at most two hedges. We refer to such hedge graphs as \emph{bi-hedge graphs}. 
By \Cref{lem:paths}, observe that \HCD remains NP-complete on bi-hedge graphs. 

 
To avoid heavy notation and without affecting consistency, 
we refer to a hedge of $\mathcal{E}(H)$ both as $x$ or $E_x$, whenever there is no ambiguity.  
For a hedge $x$, we denote by $SK_3(x)$ (single-$K_3$) the set of triangles $K_3=(a,b,c)$ in $H$ such that exactly one of the three edges belongs to $E_x$, that is $\{a,b\} \in E_x$ and $\{a,c\}, \{b,c\} \notin E_x$. Since every triangle contains at most two hedges, notice that both $\{a,c\}$ and $\{b,c\}$ belong to a single hedge $E_y \neq E_x$. 
Let $D(x)$ be the set of hedges that appear in $SK_3(x)$. 
We say that a hedge $y$ \emph{dominates} a hedge $x$ if there is a sequence of hedges 
$\langle y=h_1, \ldots, h_p=x \rangle$ with $p \geq 1$ such that $h_i \in D(h_{i+1})$, for all $1\leq i <p$. 
We denote by $R(x)$ all hedges that dominate $x$. Notice that $x\in R(x)$. 
\Cref{fig:domination} illustrates the corresponding notions.  

The intuition behind the previous definitions is that whenever the hedge $E_x$ is deleted from $H$, all hedges that dominate $x$ span new $P_3$ in $H \setminus \{E_x\}$. In fact, we will show that all such hedges do not create new $P_3$ in the hedge graph $H \setminus R(x)$. 

\begin{figure}[t]
\centering
\begin{tikzpicture}[scale=1.5]
\begin{scope}[shift={(0,0)}]
\node[draw, circle, fill=black, inner sep=1.5pt
] at (1.60,20.00) {};
\node[draw, circle, fill=black, inner sep=1.5pt
] at (2.80,20.40) {};

\node[draw, circle, fill=black, inner sep=1.5pt
] at (2.80,19.60) {};
\node[draw, circle, fill=black, inner sep=1.5pt
] at (4.00,20.00) {};
\node[draw, circle, fill=black, inner sep=1.5pt
] at (5.20,20.40) {};

\node[draw, circle, fill=black, inner sep=1.5pt
] at (5.20,19.60) {};
\node[draw, circle, fill=black, inner sep=1.5pt
] at (1.60,18.40) {};
\node[draw, circle, fill=black, inner sep=1.5pt
] at (2.80,18.80) {};

\node[draw, circle, fill=black, inner sep=1.5pt
] at (2.80,18.00) {};
\node[draw, circle, fill=black, inner sep=1.5pt
] at (4.00,18.40) {};
\node[draw, circle, fill=black, inner sep=1.5pt
] at (5.20,18.80) {};

\node[draw, circle, fill=black, inner sep=1.5pt
] at (5.20,18.00) {};
\node[draw, circle, fill=black, inner sep=1.5pt
] at (1.60,16.80) {};
\node[draw, circle, fill=black, inner sep=1.5pt
] at (2.80,17.20) {};

\node[draw, circle, fill=black, inner sep=1.5pt
] at (2.80,16.40) {};
\node[draw, circle, fill=black, inner sep=1.5pt
] at (4.00,17.20) {};
\node[draw, circle, fill=black, inner sep=1.5pt
] at (4.80,17.20) {};

\node[draw, circle, fill=black, inner sep=1.5pt
] at (4.00,16.40) {};
\node[draw, circle, fill=black, inner sep=1.5pt
] at (4.80,16.40) {};

\node[draw, circle, fill=black, inner sep=1.5pt
] at (5.60,17.20) {};
\node[draw, circle, fill=black, inner sep=1.5pt
] at (5.60,16.40) {};

\draw (1.60,20.00) -- (2.80,20.40) node[midway, above] {$x_1$};
\draw (1.60,20.00) -- (2.80,19.60) node[midway, below] {$x_1$};
\draw (2.80,20.40) -- (2.80,19.60) node[midway, right] {$x_3$};

\draw (4.00,20.00) -- (5.20,20.40) node[midway, above] {$x_3$};
\draw (4.00,20.00) -- (5.20,19.60) node[midway, below] {$x_3$};
\draw (5.20,20.40) -- (5.20,19.60) node[midway, right] {$x_4$};

\draw (1.60,18.40) -- (2.80,18.80) node[midway, above] {$x_2$};
\draw (1.60,18.40) -- (2.80,18.00) node[midway, below] {$x_2$};
\draw (2.80,18.80) -- (2.80,18.00) node[midway, right] {$x_3$};

\draw (4.00,18.40) -- (5.20,18.80) node[midway, above] {$x_3$};
\draw (4.00,18.40) -- (5.20,18.00) node[midway, below] {$x_3$};
\draw (5.20,18.80) -- (5.20,18.00) node[midway, right] {$x_5$};

\draw (1.60,16.80) -- (2.80,17.20) node[midway, above] {$y_1$};
\draw (1.60,16.80) -- (2.80,16.40) node[midway, below] {$y_1$};
\draw (2.80,17.20) -- (2.80,16.40) node[midway, right] {$y_2$};

\draw (4.00,17.20) -- (4.80,17.20) node[midway, below] {$x_2$};
\draw (4.80,17.20) -- (5.60,17.20) node[midway, below] {$z_1$};
\draw (4.00,16.40) -- (4.80,16.40) node[midway, below] {$x_5$};
\draw (4.80,16.40) -- (5.60,16.40) node[midway, below] {$y_2$};
\end{scope}

\begin{scope}[shift={(7,-5.12)}, scale=1.2]

\draw[black, fill=green!20, opacity=0.4, rotate around={0:(1.2,20.1)}] 
      (1.2,20.1) ellipse (1.25cm and 0.9cm);

\draw[black, fill=blue!15, opacity=0.4, rotate around={0:(2.4,18.75)}] 
      (2.4,18.75) ellipse (0.7cm and 0.35cm);

\draw[black, fill=red!15, opacity=0.4] 
      (0.4,18.75) ellipse (0.3cm and 0.3cm);

\node[draw, circle, fill=black, inner sep=1.5pt, label={above:$x_1$}] at (0.4,20.4) {};
\node[draw, circle, fill=black, inner sep=1.5pt, label={above:$x_3$}] at (1.2,20.0) {};
\node[draw, circle, fill=black, inner sep=1.5pt, label={above:$x_4$}] at (2.0,20.4) {};
\node[draw, circle, fill=black, inner sep=1.5pt, label={above:$x_2$}] at (0.4,19.6) {}; 
\node[draw, circle, fill=black, inner sep=1.5pt, label={above:$x_5$}] at (2.0,19.6) {};
\node[draw, circle, fill=black, inner sep=1.5pt, label={below:$y_1$}] at (2.0,18.8) {};
\node[draw, circle, fill=black, inner sep=1.5pt, label={below:$y_2$}] at (2.8,18.8) {};
\node[draw, circle, fill=black, inner sep=1.5pt, label={below:$z_1$}] at (0.4,18.8) {};

\draw[-{Stealth[length=2.1mm]}] (0.4,20.4) -- (1.2,20.0);
\draw[-{Stealth[length=2.1mm]}] (1.2,20.0) -- (2.0,20.4);
\draw[-{Stealth[length=2.1mm]}] (0.4,19.6) -- (1.2,20.0);
\draw[-{Stealth[length=2.1mm]}] (1.2,20.0) -- (2.0,19.6);

\draw[-{Stealth[length=2.1mm]}] (2.0,18.8) -- (2.8,18.8);

\draw (0.4,19.6) -- (0.4,18.8);
\draw (2.0,19.6) -- (2.8,18.8);
\end{scope}

\end{tikzpicture}
\caption{Illustrating the domination relation. On the left side, a hedge graph is shown by enumerating all $P_3$ and $K_3$ in the underlying graph. On the right side, the domination relation is depicted among the hedges by introducing a graph representation: for any two hedges $h$ and $h'$ (i) there is a directed edge from $h$ to $h'$ whenever $h \in D(h')$ and (ii) there is an edge $\{h,h'\}$ whenever there is a $P_3$ with edges spanned by $h$ and $h'$. Observe that if we remove the hedge $x_5$ then all hedges $x_1,x_2,x_3$ must be removed which constitute the hedges that dominate $x_5$, that is, $R(x_5)=\{x_1,x_2,x_3,x_5\}$. All hedges  of $R(x_5)$ reach $x_5$ with a directed path in this graph representation. 
}
\label{fig:domination}
\end{figure}
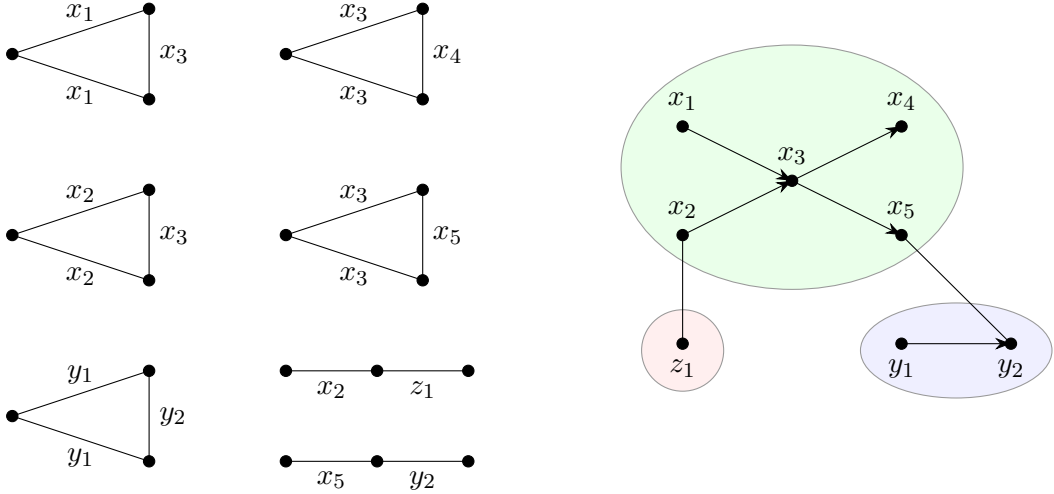

\begin{lemma}\label{lem:domproperties}
Let $H$ be a bi-hedge graph 
and let $x \in \mathcal{E}(H)$ be a hedge of $H$.
Then, the following hold. 
\begin{itemize}
    \item[(i)] For a hedge $y \in \mathcal{E}(H)$, we have $y\in R(x)$ if and only if $R(y) \subseteq R(x)$.
    \item[(ii)] Let $U$ be a solution for \HCD on $H$. If $x \in U$ then $R(x) \subseteq U$. 
    \item[(iii)] Any $P_3$ in the underlying graph of $H \setminus R(x)$ is a $P_3$ in the underlying graph of $H$. 
\end{itemize}
\end{lemma}
\begin{proof}
We start with the first claim. Assume that $y\in R(x)$ and there is a hedge $z \in R(y)$ such that $z \notin R(x)$. 
Since $z \in R(y)$, $z$ dominates $y$, which implies that there is a sequence $\langle z=h_1, \ldots, h_p=y \rangle$.
Moreover, $y\in R(x)$ implies that there is a sequence $\langle y=f_1, \ldots, f_q=x \rangle$. By combining the two sequences, we get $z$ dominates $x$ and so $z \in R(x)$ which contradicts our assumption. 
For the opposite direction, by definition $y \in R(y)$. Since $R(y)\subseteq R(x)$, we get $y \in R(x)$.

We continue with the second one. 
Assume $x \in U$. Let $\langle h_1, \ldots, h_j=y, \ldots,h_p=x \rangle$ be a sequence that dominates $x$ and let $y$ be the rightmost hedge for which holds $y\notin U$. This means $h_{j+1} \in U$.
Moreover, $y\in D(h_{j+1})$  which implies there is a $K_3 \in SK_3(h_{j+1})$ with the two other edges of $K_3$ belong to hedge $y$. 
Since $h_{j+1} \in U$, there is a $P_3$ in $H \setminus U$ that is formed by the two edges that belong to hedge $y$. 
Hence, there is a contradiction since $U$ is a solution for \HCD.

Next, assume for contradiction that there is a $P_3$ in the underlying graph of $H \setminus R(x)$ that is not a $P_3$ in the underlying graph of $H$. 
This means that the vertices of $P_3$ form a $K_3$ in the underlying graph of $H$.
Since $H$ is a bi-hedge graph, the $K_3$ belongs to $SK_3(y)$ for some $y \in R(x)$.
Thus, the other two edges of the $K_3$ belong in the same hedge, we call this hedge $z$.
Hence, $z \in D(y)$ which implies $z \in R(y)$ and, thus, by (i) we get $z \in R(x)$.
Therefore, there is no such $P_3$ in the underlying graph of $H \setminus R(x)$.    
\end{proof}

Moreover, it is not difficult to see that computing all domination relations between the hedges can be done in polynomial time. 

\begin{lemma}\label{lem:polyR}
Given a bi-hedge graph $H$, there is a polynomial-time algorithm to compute $R(x)$ for any hedge $x$ of $H$. 
\end{lemma}
\begin{proof}
We begin by enumerating all triangles of the underlying graph of $H$ in $O(n^3)$ time. 
We then compute $D(x)$ by exploring the set of triangles $SK(x)$. Then, for any triangle we know that there are at most two hedges involved, since $H$ is bi-hedge. 
Thus we can create an auxiliary directed graph (with bi-directional edges allowed) that encapsulates all relations for the set $D(x)$. 
To compute the domination relation it is enough to find all vertices that reach $x$ in the auxiliary graph by \Cref{lem:domproperties}. Such computations can be done in $O(n^3 \ell)$ time, where $\ell$ is the number of hedges in $H$. 
\end{proof}

For the approximation algorithm, we will reduce \HCD  to \VC in two steps. In particular,  we first define a variant of \VC where apart from the given graph $G$,  a set of vertex-lists $\mathcal{L}=\{L(v)\subseteq V(G) \, |\, v \in V(G)\}$ is also given and the goal is to find a minimum vertex cover $S\subseteq V(G)$ such that $\bigcup_{v \in S} L(v)=S$. For a set $X$ of vertices, we write $L(X)$ to denote $\bigcup_{x \in X} L(x)$.  
We refer to the latter problem as \MVC which is formally defined as follows. 
\vspace*{-0.01in}
\pbDef{\MVC}
	{A graph $G$, $L(v) \subseteq V(G)$ for all $v \in V(G)$,  and a non-negative integer $k$.}
	{Decide whether there is a vertex cover $S \subseteq V(G)$ such that $L(S)=S$ and $\lvert S \rvert \leq k$.}

We are not aware if \MVC can be approximated within a constant factor. 
Instead, we show it is enough to consider a restricted variation of \MVC that fulfills the following property on the vertex-lists: 
\begin{itemize}
    \item[(P1)] for every $x \in L(y)$, we have $L(x) \subseteq L(y)$.
\end{itemize}

Notice that such a property is stated within \Cref{lem:domproperties}~(i) with respect to the domination relation of the hedges. 
We now give the claimed reduction. For technical reasons, by \Cref{obs:internalP3} we can assume there is no hedge that spans an internal $P_3$ in the bi-hedge graph.

\begin{lemma}\label{lem:HCDtoMVC}
Let $H$ be a bi-hedge graph in which no hedge spans an internal $P_3$. There exists a polynomial time algorithm that, given an instance $(H,k)$ for \HCD, produces an equivalent instance for \MVC that satisfies (P1).
\end{lemma}
\begin{proof}
  Let $H=(V,\mathcal{E})$ be a bi-hedge graph. From $H$, we construct a graph $G$ and a set $\mathcal{L}$ for \MVC as follows. For every $E_x \in \mathcal{E}$, there is a corresponding vertex $x$ in $V(G)$ and there is an edge between $x$ and $y$ if the corresponding hedges $E_x$ and $E_y$ form a $P_3$ in $H$. Moreover, for every vertex $v \in V(G)$, we have $L(v)=R(E_v)$ for the corresponding hedge $E_v$. 
%
Observe that the construction takes polynomial time because $V(G)$ is exactly $|\mathcal{E}|$ and every $L(v)$ can be computed in polynomial time by \Cref{lem:polyR} and has size at most $|\mathcal{E}|$. 
Notice also that, by construction, the statements for the hedges of \Cref{lem:domproperties} hold for the corresponding vertices in $G$. Thus, property (P1) is satisfied for the vertex-lists of $G$.  
We claim that there is a solution $U$ for \HCD of size at most $k$ if and only if there is a solution $S$ for \MVC in $G$ of size at most $k$.

Let $U$ be a solution for \HCD with $k$ hedges. 
For every $P_3$, there exists a hedge $x$ such that $x \in U$. 
We claim that $S = U$ is a solution for \MVC in $G$. 
By construction, for every $P_3$ in $H$ there is an edge in $G$. Hence, we have that $S$ is a vertex cover for $G$, since there are no edges in $G - S$. 
Moreover, by \Cref{lem:domproperties}~(ii) we have that $R(E_x) \subseteq U$ for any hedge $E_x$. In terms of the vertex-lists of $G$, this means that $L(S) = S$, because $x \in L(x)$ for any vertex of $G$. 
Therefore, $S$ is a vertex cover for $G$ with $L(S) =S$ and $|S|\leq k$.

For the opposite direction, assume that there is a vertex cover $S$ of size at most $k$ for $G$ such that $L(S)=S$ and the vertex-lists satisfy property (P1).
We claim that $U = S$ is a solution for \HCD. 
Assume for contradiction that there is at least one $P_3$ in $H \setminus U$. 
Suppose that the underlying graph of $H \setminus U$ contains a $P_3=(a,b,c)$ induced by the vertices $a,b,c$. 
Let $E_x$ be a hedge of $H \setminus U$ such that $\{a,b\} \in E_x$.  
If $a,b,c$ induce a $P_3$ in $H$, then we have $\{b,c\} \in E_y$ with $E_x \neq E_y$, since there are no hedges in $H$ that span an internal $P_3$.
Thus $x$ or $y$ belongs to $S$, because $S$ is a vertex cover of $G$.
Hence, the vertices $a,b,c$ form a $K_3$ in $H$. As $H$ is bi-hedge, there are two hedges involved in the $K_3$. 
Since $\{a,c\}$ is not an edge of $H \setminus U$, the single hedge of the $K_3$ contains the edge $\{a,c\}$, so that $\{a,b,c\} \in SK(y)$. 
We conclude that $\{a,b\}, \{b,c\} \in E_x$, $\{a,c\} \in E_y$, and $E_y \in U$. 
Thus, by definition we have $x \in D(y)$, which implies $x \in R(y)$. 
This leads to a contradiction by \Cref{lem:domproperties}~(ii), since $y \in U$ and $R(y) \nsubseteq U$.
Therefore, $U$ is a solution for \HCD.
\end{proof}

We turn our attention to solve \MVC that satisfies property (P1). We take advantage of the following result. 

\begin{observation}\label{obs:capL}
For every edge $\{x,y\} \in E(G)$, we have $L(x)\cap L(y) \subseteq S$, where $S$ is any solution of \MVC. 
\end{observation}
\begin{proof}
Let $z$ be a vertex such that $z \in L(x)\cap L(y)$. Assume for contradiction that $z\notin S$.
Since $\{x,y\}$ is an edge of $G$, either $x$ or $y$ belongs to $S$. 
Without loss of generality, assume that $x\in S$.
By the fact that $L(S)=S$ and $x\in S$, we have $L(x)\subseteq S$. The latter leads to a contradiction since $z\notin S$ and $z \in L(x)$.
Hence, $L(x)\cap L(y)\subseteq S$.
\end{proof}

\noindent We say that a solution $S$ is \emph{minimal} for \MVC  if for every vertex $x \in S$, 
\begin{itemize}
    \item there is an edge $\{x,y\} \in E(G)$ with $y\notin S$, or 
    \item there is a vertex $y \in S$  such that $x \in L(y)$ and $\{y,z\} \in E(G)$ with $z\notin S$. 
\end{itemize}

\begin{figure}[t]
\centering
\begin{tikzpicture}[scale=2]

\draw[
  black,
  fill=blue!20,
  opacity=0.4,
  rotate around={0:(1.6,18.7)}
] (1.6,18.7) ellipse (1.8cm and 0.5cm);

\draw (0.4,19.4) -- (0.4,18.8);
\draw (2.8,19.4) -- (2.8,18.8);

\node[draw, circle, fill=white, inner sep=1.5pt, label={above:$\{ \underline{x_1}\}$}] at (-0.4,19.4) {};
\node[draw, circle, fill=white, inner sep=1.5pt, label={above:$\{ \underline{x_3}, x_2, x_1\}$}] at (1.4,19.4) {};
\node[draw, circle, fill=white, inner sep=1.5pt, label={above:$\{ \underline{x_4}, x_3, x_2, x_1\}$}] at (4.2,19.4) {};
\node[draw, circle, fill=white, inner sep=1.5pt, label={above:$\{ \underline{x_2}\}$}] at (0.4,19.4) {}; 
\node[draw, circle, fill=white, inner sep=1.5pt, label={above:$\{ \underline{x_5}, x_3, x_2, x_1\}$}] at (2.8,19.4) {};
\node[draw, circle, fill=white, inner sep=1.5pt, label={below:$\{ \underline{y_1}\}$}] at (2.0,18.8) {};
\node[draw, circle, fill=black, inner sep=1.5pt, label={below:$\{ \underline{y_2}, y_1\}$}] at (2.8,18.8) {};
\node[draw, circle, fill=black, inner sep=1.5pt, label={below:$\{ \underline{z_1}\}$}] at (0.4,18.8) {};

\end{tikzpicture}
\caption{This graph corresponds to the graph obtained from \Cref{fig:domination} towards the construction given in \Cref{lem:HCDtoMVC} for the \MVC problem. Every vertex $v$ is equipped with a list $L(v)$; the unique vertex with $v \in L(v)$ is shown underlined. We also highlight the difference between a minimum vertex cover (represented by the black vertices) and a  minimum multi-vertex cover (indicated by the shaded region).}
\label{fig:multi-vertex cover}
\end{figure}
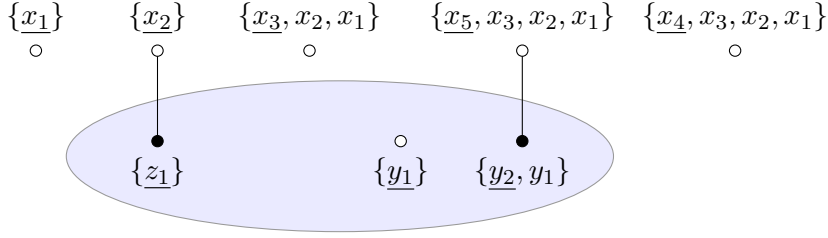
\begin{lemma}\label{lem:MVCtoVC}
If property (P1) is satisfied then there exists a 2-approximation polynomial-time algorithm for \MVC.
\end{lemma}
\begin{proof}
Given a graph $G$ and a set $\mathcal{L}$ for \MVC that satisfies (P1), we construct an auxiliary graph $G'$ as follows:
$$
V(G') = V(G)\setminus \{z \,|\, z\in L(x)\cap L(y) \, \text{ with } \{x,y\}\in E(G)\}. 
$$
For every edge $\{x,y\}\in E(G)$ with $x,y\in V(G')$, 
we add all edges between $L(x)$ and $L(y)$ to $E(G')$, that is, $E(L(x),L(y))\subseteq E(G')$. This completes the description of $G'$. 
Observe that not only the given sets $L(x)$ and $L(y)$ are vertex-disjoint as vertices that belong to $G'$, but for any two vertices $x' \in L(x)$ and $y' \in L(y)$, we know that $L(x')$ and $L(y')$ are vertex-disjoint by property (P1). 
Also notice that the subgraph of $G$ induced by the vertices $V(G')$, is a subgraph of $G'$. 
In particular, we have the following observation by construction: for every edge $\{x',y'\}\in E(G')$ there is an edge $\{x,y\}\in (E(G') \cap E(G))$ such that $x' \in L(x)$ and $y' \in L(y)$.

\begin{claim}\label{claim:subsets}
For any $x,y$ of $V(G')$ with $L(x)\subseteq L(y)$ on $G$, we have $N_{G'}(x)\supseteq N_{G'}(y)$.
\end{claim}
\begin{claimproof}
   Since $L(x)\subseteq L(y)$, we have $x \in L(y)$ by $x\in L(x)$ and property (P1). Let $z$ be a neighbor of $y$ in $G'$. By the construction of $G'$, we know that $z \in L(w)$ for some vertex $w$ of $G'$ and $\{y,w\} \in E(G)$, so that $\{y,z\} \in E(L(y),L(w))$. As $x \in L(y)$, we have that $\{x,z\} \in E(L(y),L(w))$, which concludes the claim. 
\end{claimproof}


We now map the \MVC problem on $G$ to the \VC problem on $G'$. Our goal is to prove the following: if $S'$ is an $\alpha$-approximate solution for \VC on $G'$ then there is a polynomial-time algorithm to construct an $\alpha$-approximate solution $S$ for \MVC on $G$. 
Instead of looking at any solution, we consider the \emph{minimal} ones. 
We claim that there is a minimal solution $S$ for \MVC if and only if there is a minimal solution $S'$ for \textsc{Vertex Cover} in $G'$ and such a mapping between $S$ and $S'$ can be done in polynomial time.

Let $S$ be a minimal solution for \MVC. 
We show that $S' = S \cap V(G')$ is a minimal vertex cover for $G'$.
We begin by showing that $S'$ is a vertex cover for $G'$. 
Suppose that there is an edge $\{x',y'\}$ in $G'$ such that $x',y'\notin S'$. 
Then, 
we know that there is an edge $\{x,y\} \in (E(G) \cap E(G'))$ such that $x' \in L(x)$ and $y' \in L(y)$. This means that $x$ or $y$ belongs to $S$, and, by the fact that $L(S)=S$, we have $x'$ or $y'$ belonging to $S$. Since $x',y'$ are both vertices of $G'$, we conclude that $S'$ contains at least one of them. Hence, $S'$ is a vertex cover for $G'$.
%

Next, we show the minimality for $S'$.
Let $x \in S \cap V(G')$. We show that there exists a vertex $w$ in $G'$ such that $w \notin S'$ with $\{x,w\} \in E(G')$. By construction, $x \in S$. Since $S$ is minimal, we have two cases to consider. 
Assume that there is a vertex $y \notin S$ such that $\{x,y\} \in E(G)$. 
Then $y$ is a vertex of $G'$, because for all vertices of $V(G)\setminus V(G')$ we know that they belong to $S$ by \Cref{obs:capL}. Thus $\{x,y\} \in E(G')$ by construction, which means that $S'$ is minimal because $y \notin S'$ and $x \in S$. 
Assume that there is a vertex $y \in S$ such that $x \in L(y)$ and $\{y,z\} \in E(G)$ with $z\notin S$. 
Again, $z\notin S$ implies that $z \in V(G')$. 
Now if $y \notin V(G')$ then by property (P1) we have $x \notin V(G')$, which is not possible. 
Thus $x,y,z$ are vertices of $G'$, which implies that $\{x,z\} \in E(G')$ by \Cref{claim:subsets}, since $x \in L(y)$. 
Therefore, $S'$ is a minimal vertex cover for $G'$. 


For the opposite direction, assume a minimal vertex cover $S'$ for $G'$.
We show that $S=S'\cup \{z \,|\, z\in L(x)\cap L(y) \, \text{with } \{x,y\}\in E(G)\}$ is a minimal solution of \MVC. 
By \Cref{obs:capL}, all vertices of $L(x)\cap L(y)$ with $\{x,y\} \in E(G)$ belong to $S$. 
Since $S'$ is a minimal vertex cover for $G'$, there exists a vertex $x' \in S'$ and a vertex $y' \notin S'$ for some edge $\{x',y'\} \in G'$.
By the discussion above, we know that there is an edge $\{x,y\}\in (E(G') \cap E(G))$ such that $x' \in L(x)$ and $y' \in L(y)$. 
We show that $x \in S$. To see this, observe that there is the edge $\{x,y'\} \in E(G')$ by construction. 
Since $y' \notin S'$ and $S'$ is a vertex cover, we conclude that $x\in S'$ which implies that $x \in S$. 
We further show that $y \notin S'$. 
If $y \in S'$, there is an edge $\{y,z\} \in E(G')$ such that $z\notin S'$ because $S'$ is minimal. 
Since $y' \in L(y)$, we have $\{y',z\}\in E(G')$. However, $y',z \notin S'$ leads to a contradiction, since $S'$ is a vertex cover for $G'$. 
Hence, $y \notin S'$ which means that $y \notin S$. 
To conclude, for any edge $\{x,y\}$ of $E(G)$, we have 
$x\in S$ and $y\notin S$, or 
$x\in S$ and there is a vertex $x' \in S$ such that $x \in L(x')$ and $\{x',y\} \in E(G)$ with $y \notin S$. This is exactly the definition of the minimal solution and we conclude that $S'$ is indeed a minimal solution for \MVC. 

It is known that there is a 2-approximation algorithm for \textsc{Vertex Cover} \cite{vazirani2001}. 
Let $S^*$ be a $2$-approximate solution for \VC on $G'$. 
We turn $S^*$ into a minimal vertex cover $S'$ on $G'$ by removing from $S^*$ all unnecessary vertices (that is, vertices $v$ with $N_{G'}(v)\subseteq S'$) and construct $S'$ in polynomial time. 
Since $S'\subseteq S^*$, the approximation ratio remains the same for $S'$.
From a $2$-approximate solution $S'$ for \VC on $G'$, we can compute a $2$-approximate solution $S$ for \MVC on $G$ by setting $S=S'\cup \{z \,|\, z\in L(x)\cap L(y) \, \text{ with } \{x,y\}\in E(G)\}$. Correctness follows from the previous discussion. 
Therefore, \MVC can be 2-approximated in polynomial time. 
\end{proof}

Now, we are ready to give the claimed constant-factor approximation algorithm for \HCD on bi-hedge graphs.

\begin{theorem}\label{theo:approxHCD}
    There is a polynomial-time $2$-approximation algorithm for \HCD on bi-hedge graphs.
\end{theorem}
\begin{proof}
Let $H$ be the given hedge graph. We show how to compute an $\alpha$-approximate solution $U$ for \HCD on $H$ with $\alpha=2$ by using an $\alpha$-approximation algorithm for \VC\footnote{We use $\alpha$ to denote that any $\alpha$-approximation algorithm for \VC results in an $\alpha$-approximation algorithm for our problem, as well. We state our results for $\alpha=2$, though it is known that \VC can be approximated within factor $\alpha \leq 2$ \cite{Karakostas09,vazirani2001}.}.
We begin by computing the hedges $Z$ that span internal $P_3$. 
By \Cref{obs:internalP3}, we know that all hedges of $Z$ belong to any solution. 
Given the instance $(H\setminus Z,k)$ for \HCD, 
we construct an equivalent instance $(G, \mathcal{L}, k)$ for \MVC by \Cref{lem:HCDtoMVC} in polynomial time.
Moreover, by \Cref{lem:MVCtoVC} we compute a minimal solution $S$ for \MVC on $(G, \mathcal{L}, k)$.
Then, we use the solution $S$ of \MVC to compute an equivalent solution $U$ for \HCD.
Since the solution $S$ for \MVC is an $\alpha$-approximate solution, the solution $U$ for \HCD on $(H\setminus Z,k)$ is also an $\alpha$-approximate solution by \Cref{lem:HCDtoMVC}. 
Therefore, $U \cup Z$ is an $\alpha$-approximate solution for \HCD on $H$.  
\end{proof}

\section{Acyclic hedge intersection graph}\label{sec:acyclic}
Here we consider the natural underlying structure formed by the hedges that we call \emph{hedge intersection graph} intended to represent the pattern of intersections among the hedges. Our main result is a polynomial-time algorithm whenever the hedge intersection graph is acyclic. 

Let $H=(V,\mathcal{E})$ be a hedge graph. The \emph{hedge intersection graph} of $H$ is denoted by $\mathcal{F}$ and defined as a simple graph with the following sets of vertices and edges: $V(\mathcal{F}) = \mathcal{E}$ and two vertices (that correspond to the hedges of $H$) $E_i$ and $E_j$ of $\mathcal{F}$ are adjacent if $i\neq j$ and there are two edges $e_i,e_j$ in $E(H)$ that have a common endpoint such that $e_i \in E_i$ and $e_j \in E_j$. We note that $\mathcal{F}$ can be constructed in polynomial time. Here we consider hedge graphs $H$ for which $\mathcal{F}$ is acyclic. The underlying graph on the right side in \Cref{fig:domination} is the hedge intersection graph of the considered hedge graph which happens to be in this case acyclic. 


In our algorithm we consider mainly the hedge intersection graph $\mathcal{F}$ of a hedge graph $H$. We refer to the elements of $V(\mathcal{F})$ and $E(\mathcal{F})$, as vertices (that correspond to the hedges $\mathcal{E}(H)$) and edges. For a vertex $x \in V(\mathcal{F})$, we refer to $E_x$ as the corresponding hedge of $\mathcal{E}(H)$. By definition, an edge $\{x,y\}$ of $\mathcal{F}$ is characterized as follows: there are edges $\{u,v\} \in E_x$ and $\{v,w\} \in E_y$ in $H$. This means that there is either a $P_3$ in $H$ or a $K_3$ whenever $\{u,w\}$ belongs to a hedge of $H$. 
In the latter, observe that every triangle of the underlying graph contains at most two hedges, because $\mathcal{F}$ is acyclic. Thus, $\{u,w\} \in (E_x \cup E_y)$. 

Moreover, we notice that \Cref{lem:domproperties} still applies, since $\mathcal{F}$ is acyclic. 
We say that two hedges $x$ and $y$ \emph{span a $P_3$} if there are edges $\{a,b\} \in E_x$ and $\{b,c\} \in E_y$ such that $\{a,c\} \notin E(H)$. Similarly, we say that two hedges $x$ and $y$ \emph{span a $K_3$} if there are edges $\{a,b\} \in E_x$ and $\{b,c\} \in E_y$ such that $\{a,c\} \in (E_x \cup E_y)$.

Next we consider the edges of $E(\mathcal{F})$. An edge $\{x,y\}$ of $\mathcal{F}$ belongs to exactly one of the following sets: 
\begin{itemize}
    \item $E_P$: $E_x$ and $E_y$ span a $P_3$ in $H$ but there is no $K_3$ spanned by edges of $E_x$ and $E_y$.  
    \item $E_{\triangle}$: $E_x$ and $E_y$ span a $K_3$ in $H$ but there is no $P_3$ spanned by edges of $E_x$ and $E_y$.  
    \item $E_{P,\triangle}$: $E_x$ and $E_y$ span both a $P_3$ and a $K_3$ in $H$.      
\end{itemize}

We now borrow the notation given in \Cref{sec:approx}. 
For any edge $\{x,y\}$ of $E_{\triangle}$ or $E_{P,\triangle}$, we know that $y \in D(x)$ or $x \in D(y)$ because there is a triangle with edges only of $E_x \cup E_y$. 
A vertex $x$ of $\mathcal{F}$ is called \emph{mixed vertex} if $x$ is incident to an edge $\{x,y\}$ of $E_{P,\triangle}$ such that $x \in D(y)$.
In what follows, to avoid repeating ourselves we let $U$ be a solution for \HCD on a hedge graph $H$ that admits an acyclic hedge intersection graph $\mathcal{F}$.

\begin{lemma}\label{lem:F1graph}
Let $Z$ be the set of hedges that span an internal $P_3$ in $H$ and let $X_{P, \triangle}$ be the set of all mixed vertices of $\mathcal{F}$. Then, $Z \cup X_{P,\triangle} \subseteq U$.     
\end{lemma}
\begin{proof}
By \Cref{obs:internalP3}, we know that $Z \subseteq U$. We now show that any mixed vertex $x$ belongs to $U$. 
Let $\{x,y\}$ be the edge of $E_{P,\triangle}$ incident to $x$. Since there is a $P_3$ spanned by edges of $E_x \cup E_y$, we know that $x$ or $y$ or both belong to $U$. Assume that $y \in U$. Then, we show that $x\in U$.  To see this, assume for contradiction that $x \notin U$. 
Observe that $x \in D(y)$ implies that there is a $K_3$ with two of its edges in $E_x$ and the other edge in $E_y$. Since $y\in U$ and $x \notin U$, we have an internal $P_3$ spanned by edges only of $E_x$. Therefore we conclude that $x\in U$.         
\end{proof}

According to \Cref{lem:F1graph}, we can safely remove the vertices of $Z \cup X_{P,\triangle}$ and consider the remaining graph.  
Let ${F}_1 = \mathcal{F} - (Z \cup X_{P,\triangle})$. 
As ${F}_1$ is an induced subgraph of $\mathcal{F}$, we know that ${F}_1$ remains acyclic. Moreover observe that the edges of ${F}_1$ are partitioned into the sets $E_P$ and $E_{\triangle}$. 
Since $Z \cup X_{P,\triangle}$ are hedges of $H$, $H'=H\setminus(Z \cup X_{P,\triangle})$ is a hedge-subgraph of $H$. 

We partition the vertices of ${F}_1$ as follows. Let $V_1, \ldots, V_q$ be the vertices of the connected components of ${F}_1 \setminus E_P$. Each $V_i$ is called a \emph{component} of ${F}_1$. 

\begin{lemma}\label{lem:onlyoneedge}
Between two components of ${F}_1$ there is at most one edge.      
\end{lemma}
\begin{proof}
Let $V_i$, $V_j$ be two components of ${F}_1$. 
If there is an edge incident to a vertex $x \in V_i$ and $y \in V_j$ then $\{x,y\} \in E_P$, because $(E_P, E_{\triangle})$ is a partition of the edges of $F_1$ and $V_i,V_j$ are different connected components of ${F}_1\setminus E_P$. 
Now observe that two edges $\{x,y\}$, $\{x',y'\}$ between $V_i, V_j$ with $x,x' \in V_i$ and $y,y' \in V_j$ would result in a cycle in $F_1$, since there is a path of $E_{\triangle}$ edges between $x$ and $x'$ and, similarly, there is a path of $E_{\triangle}$ edges between $y$ and $y'$. Therefore, there is at most one edge between two components of $F_1$. 
\end{proof}


We denote by $R_{H'}(x)$ all hedges that dominate $x$ in $H'$. 

\begin{lemma}\label{lem:samecomponent}
For any vertex $x$ in $F_1$ such that $x \in V_i$, we have $R_{H'}(x) \subseteq V_i$. 
\end{lemma}
\begin{proof}
Recall that all edges of $F_1$ are characterized as either in $E_P$ or in $E_{\triangle}$ and all vertices of the same component have a path with edges only from $E_{\triangle}$. We show that all hedges of $R_{H'}(x)$ are connected by edges of $E_{\triangle}$. 
If a vertex $y$ dominates $x$, then there is a sequence of vertices $\langle y=h_1, \ldots, h_p=x \rangle$ with $p \geq 1$ such that $h_i \in D(h_{i+1})$, for all $1\leq i <p$. For any $h_i \in D(h_{i+1})$ we have that $h_i,h_{i+1}$ belong to same triangle $K_3$ of $H'$. This means that there is an edge of $E_{\triangle}$ with endpoints in $h_i$ and $h_{i+1}$. Thus all vertices of the sequence have a path with edges only from $E_{\triangle}$ which implies that all vertices that dominate $x$ belong to the same component with $x$. 
\end{proof}

Now let us construct a graph $F_2$ from $F_1$ as follows: 
\begin{itemize}
    \item Remove all edges that have both endpoints in the same component. Notice that such edges constitute the set $E_{\triangle}$ in $F_1$.  
    \item For each edge $\{x,y\} \in E_{P}$ such that $x \in V_i$ and $y \in V_j$, add all edges $E(R_{H'}(x), R_{H'}(y))$. 
\end{itemize}

Observe that for any two vertices $x,y$ of the same component of $F_1$ it is not necessarily that $x \in R_{H'}(y)$ or $y \in R_{H'}(x)$. However, the following result holds for vertices that are in a domination relation. 

\begin{lemma}\label{lem:F2subset}
All vertices of the same component form an independent set in $F_2$. Moreover, for any two vertices $x,y$ of the same component of $F_2$ such that $x \in R_{H'}(y)$, we have $N_{F_2} (y) \subseteq N_{F_2}(x)$. 
\end{lemma}
\begin{proof}
By definition, there is no edge of $E_P$ that has both endpoints in the same component. Thus, by \Cref{lem:samecomponent} and the fact that the components form a partition of the vertices, all the vertices of the same component form an independent set in $F_2$. 
For the second statement, observe that for any edge $\{a,b\}$ of $F_2$, the vertices $a$ and $b$ belong to different components. To justify that $N_{F_2} (y) \subseteq N_{F_2}(x)$, we show that any neighbor $z$ of $y$ is a neighbor of $x$. 
Let $V_i$ and $V_j$ be the components of $y$ and $z$, respectively. Now observe that $\{y,z\} \in E(R_{H'}(y'), R_{H'}(z'))$, because there are vertices $y' \in V_i$ (not necessarily, $y'\neq y$) and $z' \in V_j$ (not necessarily, $z'\neq z$) such that $y \in R_{H'}(y')$ and $z \in R_{H'}(z')$. By \Cref{lem:domproperties}~(i), we have $R_{H'}(y) \subseteq R_{H'}(y')$. Since $x \in R_{H'}(y)$, we have $x \in V_i$ by \Cref{lem:samecomponent} and $\{x,z\} \in E(R_{H'}(y'), R_{H'}(z'))$. 
Therefore, $z$ is a neighbor of $x$, as required. 
\end{proof}

The connection between the constructed graph $F_2$ and the hedge graph $H$ is given in the following result. We note that a vertex cover of $F_2$ that is not minimal does not necessarily correspond to a cluster hedge-subgraph of $H'$. 
\begin{lemma}\label{lem:F2vertexcover}
For any minimal solution $U'$ for \HCD on $H'$, there is a minimal vertex cover $S$ of $F_2$ such that $S = U'$. Conversely, for any minimal vertex cover $S$ of $F_2$ there is a minimal solution $U'$ for \HCD on $H'$ such that $U'=S$. 
%
\end{lemma}
\begin{proof}
By construction, note that the hedges of $H'$ are in 1-to-1 correspondence with the vertices of $F_2$. Also observe that for any edge $\{x,y\}$ of $F_2$ we have the following: 
\begin{itemize}
    \item $x \in V_i$, $y \in V_j$ (all edges are between components) and 
    \item there are vertices $x' \in V_i, y' \in V_j$ such that there is an edge $\{x',y'\} \in E_P$ in $F_1$ and $x \in R_{H'}(x'), y \in R_{H'}(y')$.   
\end{itemize}
Notice that the vertices $x', y'$ between the components $V_i$ and $V_j$ are unique, due to \Cref{lem:onlyoneedge}. Moreover, due to \Cref{lem:F1graph}, there is no hedge in $H'$ that spans an internal $P_3$ in $H'$. 

Let $U'$ be a minimal solution for \HCD on $H'$. We show that $S=U'$ is a minimal vertex cover of $F_2$. 
Assume, for contradiction, that there is an edge $\{x,y\}$ in $F_2$ such that $x,y \notin S$. By the previous properties on the edges of $F_2$, we know that there are vertices $x' \in V_i, y' \in V_j$ such that there is an edge $\{x',y'\} \in E_P$ in $F_1$ and $x \in R_{H'}(x'), y \in R_{H'}(y')$. 
Since $\{x',y'\} \in E_P$, $E_{x'}$ and $E_{y'}$ span a $P_3$ in $H'$. Thus $x'$ or $y'$ belong to $U'$, because $H' \setminus U'$ is a cluster hedge-subgraph. Without loss of generality, suppose that $x' \in U'$. By \Cref{lem:domproperties}~(ii), we have $R_{H'}(x') \subseteq U'$. Then, however, we reach a contradiction to $x \in S$, since $x \in R_{H'}(x')$. 
Therefore, $S=U$ is a minimal vertex cover of $F_2$.   

Let $S$ be a minimal vertex cover of $F_2$. We show that $U'=S$ is a minimal solution for \HCD on $H'$. Clearly, $H' \setminus U'$ is a hedge-subgraph of $H'$. To show that $H' \setminus U'$ is indeed a cluster hedge-subgraph, assume for contradiction that its underlying graph contains a $P_3=a,b,c$ induced by the vertices $a,b,c$.  
Let $E_x$ be a hedge of $H'$ such that $\{a,b\} \in E_x$. Then $E_x \notin U'$. 

Suppose that the $P_3$ is induced in $H'$. If $\{b,c\} \in E_x$ then $E_x$ spans an internal $P_3$ which is not possible by \Cref{lem:F1graph}. Thus $\{b,c\} \in E_y$ with $E_x \neq E_y$ and $E_y \notin U'$. 
This means that there is an edge $\{x,y\} \in E_P$, so that $\{x,y\}$ is an edge in $F_2$, by construction. Since $S$ is a minimal vertex cover of $F_2$, we have $x\in S$ or $y\in S$ which contradicts the existence of the $P_3$. 

Next suppose that $P_3$ is induced in $H' \setminus U'$, but not induced in $H'$. This particularly means that $\{a,c\}$ being an edge in $H'$. Since $E_x \notin U'$, we have $\{a,c\} \in E_y$ with $E_y \in U'$. 
If $\{b,c\} \notin E_x$ then $\{b,c\} \in E_y$ because there is a triangle in $H'$ and $F_1$ is acyclic. This, however, contradicts the fact that $E_y \in U'$. 
Thus $\{b,c\} \in E_x$. Now observe that $E_x$ and $E_y$ span a $K_3$ in $H$ which implies that $\{x,y\}$ is an edge of $E_{\triangle}$ or $E_{P,\triangle}$. If $\{x,y\} \in E_{P,\triangle}$ then $x\in D(y)$, so that $x \in X_{P,\triangle}$. This, however, implies that $x$ is not a vertex of $F_2$ due to \Cref{lem:F1graph}. 
Thus $\{x,y\} \in E_{\triangle}$. 
In particular, $x$ dominates $y$ and both vertices belong to the same component of $F_2$ by \Cref{lem:samecomponent}. 
Since $S$ is a minimal vertex cover and $y \in S$, there is a vertex $z \notin S$ such that $\{y,z\}$ is an edge of $F_2$. Due to \Cref{lem:F2subset}, there is an edge $\{x,z\}$ in $F_2$. Then, however we reach a contradiction to the fact that $S$ being a vertex cover of $F_2$ because $x,z \notin S$. 
Therefore, $U'=S$ is a minimal solution for \HCD on $H'$.   
\end{proof}

In order to compute a minimum vertex cover of $F_2$ in polynomial time, we take advantage of the following.  

\begin{lemma}\label{lem:F2bipartite}
The graph $F_2$ is bipartite. 
\end{lemma}
\begin{proof}
Since $F_1$ is an induced subgraph of an acyclic graph $\mathcal{F}$, $F_1$ is acyclic and, thus, bipartite. 
Let $V_1, \ldots, V_q$ be the components of $F_1$ and let $F'_1$ be the subgraph of $F_1$ that does not contain any edge between vertices of the same component $V_i$. 
Let $(V_A,V_B)$ be a bipartition of the vertices of $F'_1$ where all vertices of the same component belong to $V_A$ or $V_B$. That is, for any $V_i$, we have $V_i \subseteq V_A$ or $V_i \subseteq V_B$. This can be achieved by simply contracting each component into a single vertex, take any bipartition, and then uncontract the components.
We claim that $(V_A,V_B)$ is a bipartition of the vertices of $F_2$. 
For this, observe that by \Cref{lem:F2subset} all vertices of $V_i$ form an independent set 
and for any vertex $x \in V_i$ we have $R_{H'}(x) \subseteq V_i$ by \Cref{lem:samecomponent}. 
Thus all edges of $F_2$ do not have endpoints in the same component $V_i$. 
Moreover, by the construction of $F_2$, whenever there are edges incident to vertices of $E(R_{H'}(x), R_{H'}(y))$, we know that $x \in V_i$, $y\in V_j$ such that $\{x,y\} \in E_P$. Thus, $V_i \in V_A$ and $V_j \in V_B$ or $V_j \in V_A$ and $V_i \in V_B$, since all edges of $E_P$ belong to $F'_1$. 
Therefore, $F_2$ admits a bipartition $(V_A, V_B)$ of its vertices.  
\end{proof}

Now we are ready to obtain our main result of this section, namely a polynomial-time algorithm for \HCD whenever the hedge intersection graph is acyclic.

\begin{algorithm}
\caption{A minimum set of hedges for the \HCD problem}
\label{algorithm:hedgeintersect}
	\renewcommand{\algorithmicrequire}{\textbf{Input:}}
	\renewcommand{\algorithmicensure}{\textbf{Output:}}
	\begin{algorithmic}[1]
    \Require a hedge graph $H$ that admits an acyclic hedge intersection graph
    \Ensure a minimum set of hedges $U$ of $H$ s.t. $H \setminus U$ is a cluster hedge-subgraph
\State Compute the hedge intersection graph $\mathcal{F}$ of $H$
\State Partition the edges of $\mathcal{F}$ into the three sets: $E_P$, $E_{\triangle}$, $E_{P,\triangle}$
\State Apply \Cref{lem:polyR} to compute $R(x)$ for any hedge $x$ of $H$
\State Let $Z$ be the hedges that span internal $P_3$ in $H$
\State Let $X_{P,\triangle}$ be the mixed vertices of $\mathcal{F}$
\State Compute the graph $H' = H \setminus (Z \cup X_{P,\triangle})$ and update accordingly $R(x)$
\State Compute the components $V_1, \ldots, V_q$ of $F_1 = \mathcal{F} -(Z \cup X_{P,\triangle})$ 
\State Construct the graph $F_2$ from $F_1$ as follows: 
\State \ \ \ \ remove all edges $\{x,y\}$ with $x,y \in V_i$
\State \ \ \ \ for each edge $\{x,y\} \in E_{P}$ with $x \in V_i$ and $y \in V_j$, add all edges $E(R_{H'}(x), R_{H'}(y))$
\State Let $U'$ be a minimum vertex cover of $F_2$
\State \textbf{Return} $U = Z \cup X_{P,\triangle} \cup U'$
    \end{algorithmic}
\end{algorithm}

\begin{theorem}\label{theo:acyclic}
Let $H$ be a hedge graph that admits an acyclic hedge intersection graph. Then, there is a polynomial-time algorithm for \HCD on $H$.  
\end{theorem}
\begin{proof}
We describe the steps in \Cref{algorithm:hedgeintersect}. Given $H$, we begin by computing the hedge intersection graph $\mathcal{F}$ in $O(\ell m + n)$ time, where $\ell, m, n$ stand for the number of hedges, edges, and vertices, respectively. We also enumerate all $P_3$ and $K_3$ of the underlying graph in $O(n^3)$ time. 
Having the enumerated subgraphs, we characterize each edge of $\mathcal{F}$ as $E_P$, $E_{\triangle}$, or $E_{P,\triangle}$. For the domination relation between the hedges, we can compute $R(x)$ in polynomial time by \Cref{lem:polyR}.  

We collect all hedges that span an internal $P_3$ and all mixed vertices that correspond to $Z \cup X_{P, \triangle}$. We compute the graph $F_2$ in $O(\ell^2)$ time and then compute a minimum vertex cover $U'$ of $F_2$ in polynomial time \cite{matching80,flow13}, since $F_2$ is a bipartite graph by \Cref{lem:F2bipartite}. We finally return the set $U' \cup Z \cup X_{P,\triangle}$. Correctness follows from \Cref{lem:F1graph,lem:F2vertexcover}. Therefore, all steps can be carried out in polynomial time. 
\end{proof}

\section{Concluding remarks}
In this paper, we introduced the \HCD problem. 
Our analysis for the complexity of the problem reveals that in almost every non-trivial underlying structure of a hedge graph the problem remains difficult. 
Aiming towards a poly(OPT)-approximation, we showed that the problem is \MHD-complete. From the parameterized perspective, we showed that the variant of the problem parameterized by the solution size does not admit polynomial kernel under reasonable hierarchy assumptions.
On the positive side, we studied the hedge intersection graph. We have proposed a constant factor approximation algorithm whenever at most two hedges appear in any triangle. Moreover, for any acyclic hedge intersection graph we have given a polynomial-time algorithm. We believe that our results reveal an interesting path for future work on hedge optimization problems. 


We conclude with some additional remarks and open problems related to our results. It is an open problem whether any \MHD-complete problem is actually poly-APX-complete \cite{KhannaSTW01,BliznetsCKP18}. The \MHD-completeness of \HCD adds a natural deletion problem into the class of such unclassified problems. 
Moreover, interesting connections arise with respect to parameterized complexity. We are not sure whether the parameterized variant of \emph{every} \MHD-complete problem does not admit a polynomial kernel (under the assumption NP $\nsubseteq$ coNP/poly). 
Towards such a direction, one may need to turn the $A$-reduction into a ppt-reduction that preserves the approximation guarantee into a polynomial-size kernel.  
However, we are able to show that the optimization variant of \PropSat{$f$} is indeed \MHD-complete. We are not aware of any approach concerning the approximation complexity of \PropSat{$f$} in terms of an optimal solution. Thus the following result interestingly settles a natural lower bound related to the approximation guarantee. 

\begin{theorem}\label{theo:PROPisMHD}
\PropSat{$f$} is \MHD-complete under $A$-reductions. 
\end{theorem}
\begin{proof}
By \Cref{theo:mddcompleteness}, it is enough to show
$A$-reductions from \HCD to \PropSat{$f$} and from \PropSat{$f$} to \HCD. 
The latter is given in the ppt-reduction of \Cref{theo:nokernel} with the propagational function $f=f_1$, which serves as an $A$-reduction as well. Note that the described reduction is a polynomial-time reduction, hence an $A$-reduction, from \PropSat{$f_1$} to \HCD, as the sizes of the corresponding solutions are exactly the same.  
For the other direction, we employ a proof similar to \Cref{lemma:hcdtomones}. 

In particular, we show that there is an $A$-reduction from \HCD to \PropSat{$f'$}, where $f'(x,y,z) = \lnot x \lor y \lor z$. It is not difficult to see that $f'$ is indeed a propagational function (see also \cite{CaiC15}). 
Let $H$ be a hedge graph given as an instance of \HCD, where every connected component of $H$ contains exactly three vertices by \Cref{lem:allPandK}. We show how to construct a formula $\phi$ from $H$ such that every clause of $\phi$ is of the form $f'(x,y,z)$. 
We correspond every hedge in $H$ to exactly one Boolean variable in $\phi$.  
For every $P_3$ or $K_3$ of $H$, we add a constant number of clauses in $\phi$ as explained below. 
Let $X, Y, Z$ be the hedges that appear in each component. 
\begin{itemize}
    \item For every $K_3$ whose three edges belong to three distinct hedges $X,Y,Z$, we add three clauses $f'(x,y,z), f'(y,x,z), f'(z,x,y)$ in $\phi$. It is not difficult to see that $f_1(x,y,z) = f'(x,y,z) \land f'(y,x,z) \land f'(z,x,y)$, where $f_1$ is the function defined in \Cref{sec:incompres}.  
    \item For every $K_3$ where two of its edges belong to hedge $X$ and the other edge belongs to $Y$, we add the clause $f'(y,x,0)$ in $\phi$. This can be translated as follows: whenever $y$ is set to $1$ it also enforces $x$ to $1$; otherwise, $y=0$ and $f'(y,x,0)$ is satisfied, regardless of the assignment for $x$.   
    \item For every $P_3$ whose two edges belong to two distinct hedges $X,Y$, we add the clause $f'(1,x,y)$ in $\phi$. This is justified by the fact that at least one of $x$ or $y$ must be set to $1$. 
    \item For every $P_3$ where both of its edges belong to the same hedge $X$, we add the clause $f'(1,0,x)$ in $\phi$. As explained in \Cref{obs:internalP3}, $x$ must be set to $1$. 
\end{itemize}
Note that if there is a $K_3$ whose edges belong to a single hedge then we ignore such component. This completes the construction of $\phi$ which can be done in polynomial time. 
Similarly to the proof of \Cref{lemma:hcdtomones}, we conclude that for any subset $U$ of hedges in $H$, $H\setminus U$ is a cluster hedge-subgraph if and only if the one-assignment on the corresponding Boolean variables of $U$ satisfies $\phi$. 
Therefore, the polynomial-time transformation serves as an $A$-reduction, as claimed.
\end{proof}


Another interesting point for future work is to consider the sizes of the hedges (i.e., the number of edges in each hedge) instead of the total number of hedges, as a natural candidate for the objective of the problem.  
Then notice that within the corresponding extremities of the suggested objective the behavior is well-understood: if the number of edges in each hedge is large enough then the total number of hedges is rather small, which means that the problem can be solved efficiently; otherwise, if each hedge contains exactly one edge then the problem coincides with \CD. This naturally suggests to exploit the complexity with respect to the variety on the sizes of the hedges. 
Furthermore, intriguing questions related to the complexity of \HCD arise when considering structural parameters of the hedge intersection graph such as treewidth or cliquewidth. Our algorithm on acyclic hedge intersection graphs can be seen as a first step towards such an approach.   

\bibliographystyle{plain}
\bibliography{CD_general}

\end{document}